\documentclass[sn-basic]{sn-jnl}% Basic Springer Nature Reference Style/Chemistry Reference Style
\usepackage{float}
\usepackage{latexsym}
\usepackage{amsmath} % extensions for typesetting of math
\usepackage{amsfonts}       % math fonts
\usepackage{amsthm} % theorems, definitions, etc.
\usepackage{bbding}         % various symbols (squares, asterisks, scissors, ...)
\usepackage{bm}             % boldface symbols (\bm)
\usepackage{graphicx}    
\usepackage{caption} % embedding of pictures
\usepackage[labelfont=bf]{caption}
\usepackage{caption}
\DeclareCaptionLabelSeparator{twospace}{\ ~}   %%no punctuation 
\usepackage{caption}
\DeclareCaptionLabelSeparator{twospace}{\ ~}   %%no punctuation 
\usepackage{subcaption}
\usepackage{fancyvrb}       % improved verbatim environment
\usepackage[nottoc]{tocbibind} % makes sure that bibliography and the lists
\usepackage{dcolumn}        % improved alignment of table columns
\usepackage{booktabs}       % improved horizontal lines in tables
\usepackage{paralist}       % improved enumerate and itemize
\jyear{2021}

\newtheorem{myrem}{Remark}
\newtheorem{mythm}{Theorem}

\newtheorem{mycol}{Corollary}
\newtheorem{mycon}{Conjecture}

\begin{document}

\title[Analysing the Effect of Test-and-Trace Strategy in an SIR Epidemic Model]{Analysing the Effect of Test-and-Trace Strategy in an SIR Epidemic Model}


\author*[1]{\fnm{Dongni} \sur{Zhang}}\email{dongni.zhang@math.su.se}

\author[1]{\fnm{Tom} \sur{Britton}}\email{tom.britton@math.su.se}

\affil[1]{\orgdiv{Department of Mathematics}, \orgname{Stockholm University}, \orgaddress{\city{Stockholm}, \country{Sweden}}}

\abstract{Consider a Markovian SIR epidemic model in a homogeneous community. To this model we add a rate at which individuals are tested, and once an infectious individual tests positive it is isolated and each of their contacts are traced and tested independently with some fixed probability. If such a traced individual tests positive it is isolated, and the contact tracing is iterated. This model is analysed using large population approximations, both for the early stage of the epidemic when the "to-be-traced components" of the epidemic behaves like a branching process, and for the main stage of the epidemic where the process of to-be-traced components converges to a deterministic process defined by a system of differential equations. These approximations are used to quantify the effect of testing and of contact tracing on the effective reproduction numbers (for the components as well as for the individuals), the probability of a major outbreak, and the final fraction getting infected. Using numerical illustrations when rates of infection and natural recovery are fixed, it is shown that Test-and-Trace strategy is effective in reducing the reproduction number. Surprisingly, the reproduction number for the branching process of components is not monotonically decreasing in the tracing probability, but the individual reproduction number is conjectured to be monotonic as expected. Further, in the situation where individuals also self-report for testing, the tracing probability is more influential than the screening rate (measured by the fraction infected being screened).}




\maketitle

\section{Introduction}
\label{sec:intro}

An important reason for modelling the spread of infectious diseases lies in better understanding the effect of various preventive measures, such as lockdown, social distancing, contact tracing, testing, self-isolating and quarantining. During the ongoing pandemic of Covid-19 the Test-and-Trace (TT) strategy has received a lot of attention
%%%%%%%%%%%%%%%%%%%%%%%%%%%%%%%%%%%%%%%%%%%%%%%%%%%%%%%%%%%%%%%%
%\textbf{add a couple of applied references such as those below}
%%%%%%%%%%%%%%%%%%%%%%%%%%%%%%%%%%%%%%%%%%%%%%%%%%%%%%%%%%%%%%%%
\citep{kendall_epidemiological_2020,lucas_engagement_2020, bradshaw_bidirectional_2021}. The test-part of the strategy means that testing (of suspected cases and/or randomly chosen individuals) is increased, with the hope that finding infected individuals quickly and isolating them will reduce the transmission. The trace-part of the strategy is that individuals who are tested positive are quickly questioned about their recent contacts, and such contacts are then localized, tested and isolated if testing positive. 

Contact tracing and its effect has been studied both from a theoretical perspective, and during the ongoing Covid-19 pandemic also from an applied point of view, including procedures for estimation of model parameters. For Covid-19 it has been observed in several countries that contact tracing is a highly powerful intervention measure. For instance, in the UK study \citep{kendall_epidemiological_2020} of the TT-program first carried out on the Isle of Wight, they concluded that the number of new confirmed Covid-19 cases decreased more sharply after the TT-intervention.  In \citep{lucas_engagement_2020}, they found that it is unlikely for strict self-isolation policies to improve the effectiveness by means of contact tracing. In particular, the effect of contact tracing on controlling the Covid-19 epidemic has been studied based on simulation models in several papers (e.g. \citep{di_domenico_impact_2020,firth_using_2020, keeling_efficacy_2020, kretzschmar_2020}). One benefit of simulation models is that the results are easier to interpret, whereas one shortcoming is that they are difficult to be analysed analytically. Our paper is concerned with rigorous large population approximations for a stochastic epidemic model using theory for branching processes.

%\add{In most of the frameworks addressing models for contact tracing, there are three types of contact tracing: forward tracing searches those who are potentially infected by the index case \citep{ball_threshold_2011,ball_stochastic_2015,muller_contact_2000}; backward tracing looks for who infected the index case \citep{muller_contact_2000}; full tracing stands for tracing both backward and forward \citep{bradshaw_bidirectional_2021,muller_contact_2000}.Beside these types of tracing, \mbox{\citep{sideward_ct}} suggests that there is a "sideward" tracing when tracing the "superspreading events" (where several individuals are simultaneously infected). They show that these large gatherings could however improve the effectiveness of contact tracing. }

More theoretical studies often make simplifying assumptions in order to make more analytical progress. For example, under the assumption of a homogeneous mixing population, \citep{ball_threshold_2011,ball_stochastic_2015} consider the traditional SIR model with forward tracing (where either a fraction of the infectees of a parent case is tested or none of the contacts is reported) but without tracing backwards to infectors of tested individuals. The model in \citep{bradshaw_bidirectional_2021} suggests that both backward and forward tracing would remarkably increase the effectiveness of Covid-19 epidemic control.
\citep{muller_contact_2000} deals with a stochastic SIRS model among a homogeneous mixing population. Once infectious individuals are discovered, each of their possible infectious contacts will be traced and treated (equivalent to isolated) with some probability. It is analysed (with focus on the age since infection) gradually in three cases, namely backward tracing, forward tracing and tracing both ways. They derive the critical tracing probability for reducing the effective reproduction number to below 1 so that a major outbreaks no longer may occur.
\citep{ball_threshold_2011} is concerned with a SIR epidemic model in a homogeneous mixing population.
Diagnosed individuals in \citep{ball_threshold_2011} are asked to name a fraction of their infectees, who will be isolated (i.e. removed) immediately if they have not been diagnosed earlier. Further, those traced individuals are asked to name their infectious contacts in the same way, otherwise none of their contacts will be named. The model studied in \citep{ball_stochastic_2015} extends the one in \citep{ball_threshold_2011} by introducing the exposed period and tracing delays, where the infectious individuals who share the same infector are traced after independent delay times. It is also assumed that untraced individuals may not be asked to name their infectees, for instance when they are asymptomatic. Numerical results in \citep{ball_stochastic_2015} indicate that independent delay times have bigger effect on the spreading compared to the situation where the delay times from one infected individual being contact traced are all the same. Recently, 
\citep{muller_contact_2021} adds super-spreader events, where several 
individuals may get infected by one infector at the same time, to a model having contact tracing similar to \citep{muller_contact_2000}.

Beside forward and backward tracing, \mbox{\citep{sideward_ct}} suggests that there is a "sideward" tracing when tracing large gatherings (where infected asymptomatic individuals could be traced even if they are neither the infectees nor the infector of the index case). Additionally, \citep{barlow_branching_2020} analyse the epidemic model as a branching process with contact tracing on top of its genealogical tree. According to their assumption, each infective will be detected with some probability after a certain number of generations and then each of its contacts will be successfully traced with some probability. 
As in our paper, they focus on the evolution of "traceable clusters" (in our paper called "to-be-reported components"), but from a different point of view: they extend the percolation-based analysis to contact tracing and give the approximated expression of the probability of extinction.

In the present paper we consider an SIR model with TT strategy in which not all individuals necessarily are traced, tracing is both backward and forward, but on the other hand assuming no latent periods and no delay before contact tracing happens. 
To conclude, we study a stochastic SIR epidemic model including Test-and-Trace prevention for a large finite population. The Test-feature is modelled by assuming that infectious individuals are tested (screened) at a constant rate, and for tracing we assume that an individual who tests positive reports each contact independently and reported contacts are traced and tested without delay. In the tracing procedure we assume that 
currently infectious as well as those who by now have recovered are 
identified, and that the tracing procedure is iterated for both categories.

More precisely, we analyse a homogeneous SIR epidemic model having four parameters: the rate of infectious contacts $\beta$, the recovery rate $\gamma$, the testing rate $\delta$ for infectious individuals, and the fraction $p$ of contacts that are reached in the contact tracing procedure. Using large population approximations, we analyse both the initial phase of the epidemic (where it behaves like a certain branching process) and the main phase of the epidemic when it can be approximated by a deterministic process.
The main focus of the paper is to shed light on how much is gained from the TT-strategy. For example, how should resources optimally be distributed between testing and contact tracing, how much would the reproduction number be reduced for achievable levels in the TT-strategy, and so on.

In Section \ref{sec:main-results} we present the model details and our main results and some intuitions for how the results are obtained. In Section \ref{sec:earlystage} and \ref{sec:branchingprocess} we give more details and proofs to the analyses of the initial stage of the epidemic and its main phase in Section \ref{sec:mainphase}. In Section~\ref{sec:numerical} we report simulations and numerical studies of the model and study how effective the TT-strategy is. The paper ends with a conclusion where extensions and possible improvements are discussed.

\section{Model and Main Results}
\label{sec:main-results}
%%%%%%%%%%%%%%%%%%%%%%%%%%%%%%%%%%%%%%%%%%%%%%%%%%%%%%%%
%%%%%%%%%%%%%% Standard SIR Model %%%%%%%%%%%%%%
%%%%%%%%%%%%%%%%%%%%%%%%%%%%%%%%%%%%%%%%%%%%%%%%%%%%%%%%
\subsection{The Standard SIR Model }\label{sec:model_SIR}
We start with a Markovian \textbf{SIR} (\textbf{S}usceptible $\rightarrow$
\textbf{I}nfectious
$\rightarrow$
\textbf{R}ecovered) epidemic spreading in a closed and homogeneous mixing population. By closed, we mean that there is no influx of new susceptibles or death. At any time point, each individual is either susceptible, infectious or recovered. We assume that the size of population is $n$, and that initially one individual is infectious individual and the rest are susceptible. Each susceptible becomes infectious once he/she makes contact with an infective. Times at which such contacts between two given individuals occurs are constructed by a homogeneous Poisson process with rate $\beta/n$. Equivalently, an infectious individual has %infectious 
contacts at rate $\beta$, each time with a uniformly selected individual each thus having probability $1/n$. Only contacts with susceptibles result in infection whereas other contacts have no effect. Once an individual gets infected, he/she remains infectious for a random time $T_{I}.$ We assume that the period $T_{I}$ is independent, exponentially distributed with mean $\mathbb{E}[T_{I}]=1/\gamma$, and the parameter $\gamma$ denotes the rate of natural recovery. So the underlying model is Markovian. Once naturally recovered, the individual plays no role in the spreading of epidemic. The epidemic stops when there is no infectives.
%%%%%%%%%%%%%%%%%%%%%%%%%%%%%%%%%%%%%%%%%%%%%%%%%%%%%%%%
%%%%%%%%%%%%%% Testing and Tracing System %%%%%%%%%%%%%%
%%%%%%%%%%%%%%%%%%%%%%%%%%%%%%%%%%%%%%%%%%%%%%%%%%%%%%%%
\subsection{The Markovian SIR-TT model}\label{sec:model_TT}
Now we incorporate our Test-and-Trace scheme into this SIR model. It is additionally assumed that infectious individual are tested at rate $\delta$ (possibly also non-infectious individuals are tested at this rate but this has no effect and is hence not assumed). Individuals that test positive are called diagnosed and diagnosed individuals are immediately isolated thus not taking further part in disease spreading. So, infectious individuals can stop spreading disease either from natural recovery (rate $\gamma$) or from being tested and diagnosed (rate $\delta$). Individuals that are diagnosed are also contact traced. This is modelled by assuming that a diagnosed individual reports each of its infectious contacts (both the infector and infectees) independently with probability $p$. The individuals that are traced in this way are tested, and individuals that test positive (either still being infectious or by then having recovered) are then contact traced in the same way (so contact tracing is iterated among those that have been infected). To simplify modelling we assume no delay in this contact tracing and instead assume that it all happens instantaneously. 

The SIR-TT model makes two simplifying assumptions, that 
contact tracing occurs without delay, and that also traced individuals 
who have by now recovered are contact traced. In reality tracing certainly takes some time, and individuals who have recovered several days 
or even weeks earlier would typically not be contact traced. The results 
from the present model can hence serve as an upper bound on how 
effective "real" contact tracing may be.
All contact and reporting processes as well as infectious periods are defined mutually independent. 
In Table ~\ref{tab:para} we list all the model parameters.

It is possible to consider different models for how different individuals would report in relation to each other. For instance, would an infector $A$ report an infectee $B$ independent of whether the infectee $B$ would report the infector $A$? A ''yes'' to the answer could for example happen if what defines a ''contact'' is not clarified enough so what $A$ considers as a contact may not coincide what $B$ thinks, and a ''no'' could happen if some of the contacts are with friends/acquaintances and other contacts are between unknown people, e.g.\ on the bus. However, it is clear from the model description that once $A$ or $B$ are diagnosed and asked to report their contacts the reporting event in the opposite direction is useless. This is true also if the first reporting event resulted in not naming the other individual: a later contact tracing of that individual has no effect on the first individual since he/she has already been diagnosed. As a consequence, all models for how contacts report each other (independently, symmetrically or some partial dependence) will result in exactly the same stochastic model. In our description below we have chosen to use the symmetric description thus assuming that A reports B if and only if B reports A, but this is only for practical purposes.

%%%%%%%%%%%%%%%%%%%%%%%%%%%%%%%%%%%%%%%%%%%%%%%%
%%%%%%%%%%%%%%%% alternative model%%%%%%%%%%%%%%%%
%%%%%%%%%%%%%%%%%%%%%%%%%%%%%%%%%%%%%%%%%%%%%%%%
Our model assumes that infectious individuals either recover naturally (at rate $\gamma$) or are tested and diagnosed at rate $\delta$. There is an \textbf{alternative model interpretation}, which is more detailed in the sense that testing could take place also prior to screening. In addition to those who are found by screening, some infectious individuals may test themselves, e.g.\ due to symptoms. This scenario also fits into the present model by simply adding one more parameter. The parameters $\gamma$ and $\delta$ are unchanged: $\gamma$ is the rate of natural recovery and $\delta$ denotes the testing rate (screening). But now we add a rate $\nu$ at which infectious individuals self-report and test themselves. Both self-reporting and screening trigger contact tracing, so all that matters for the epidemic spreading is the sum $\nu+\delta$ of these two rates. As a consequence, this new model interpretation with 5 parameters $(\beta, \gamma, \delta, \nu, p)$ is identical to the original model with the following 4 parameter values $(\beta, \gamma, \delta + \nu, p)$. Since the alternative model interpretation fits into the original model all mathematical results from the original model apply. In Section \ref{sec:numerical} we give some numerical results also for the alternative model.
%%%%%%%%%%%%%%%%%%%%%%%%%%%%%%%%%%%%%%%%%%%%%%%%%%%%%%%%
%%%%%%%%%%%%%% Table of parameters %%%%%%%%%%%%%%
\begin{table}
\caption{Table with all model parameters}
\label{tab:para}       % Give a unique label
\centering
\begin{tabular}{llll}
\hline\noalign{\smallskip}
Parameter & Notation \\ 
\noalign{\smallskip}\hline\noalign{\smallskip}
Infection rate & $\beta$  \\ 
Rate of natural recovery & $\gamma$ \\
Rate of diagnosis & $\delta$ \\
Tracing probability & $p$ \\
Size of population & $n$ \\
\noalign{\smallskip}\hline

\end{tabular}
\end{table}
%%%%%%%%%%%%%%%%%%%%%%%%%%%%%%%%%%%%%%%%%%%%%%%%%%%%%%%%%%%%%%%%%%%%%%
%%%%%%%%%%%%%% Symmetric and Independent Reporting Contacts %%%%%%%%%%%%%%
%%%%%%%%%%%%%%%%%%%%%%%%%%%%%%%%%%%%%%%%%%%%%%%%%%%%%%%%%%%%%%%%%%%%%%
\subsection{Main Results}

%%%%%%%%%%%%%%%%%%%%%%%%%%%%%%%%%%%%%%%%%%%%%%%%%%%%%%%%%%%%%%%%%
%%%%%%%%%%%%% coupling method , branching process %%%%%%%%%%%%%
%%%%%%%%%%%%%%%%%%%%%%%%%%%%%%%%%%%%%%%%%%%%%%%%%%%%%%%%%%%%%%%%%

We start by considering the beginning of the 
epidemic where we prove that the epidemic, asymptotically as the 
population size grows to infinity, converges to a certain limit 
process.

We assume that there is one alive ancestor at time zero. Each alive individual gives birth at rate $\beta$, dies naturally with rate $\gamma$ and is removed from the population with rate $\delta$ (corresponding to being tested and diagnosed in the epidemic). Further, an individual who is removed also leads to that each of its offspring as well as its parent will be removed independently with probability $p$. All those that are removed in this step will in turn lead to that its parent and offspring will be removed independently with probability $p$ and so on. This limit process is identical to the SIR-TT epidemic process defined in Section \ref{sec:model_TT} (denoted by $E_n(\beta,\gamma,\delta,p)$) with one single exception. In the epidemic model an infectious individual infects new individuals at rate $\beta(S_n(t)/n)$ where $S_n(t)$ denotes the number of susceptibles at $t$, since only contacts with susceptibles result in infection and this has probability $S_n(t)/n$. On the other hand, in the limit process alive individuals give birth at constant rate $\beta$. Nevertheless, in the beginning and assuming a large population then these two rates will be close to each other since then $S_n(t)\approx n$.

The contact tracing mechanism induces a dependence 
between individuals both in the epidemic as well as the limiting 
process. Rather than studying \emph{individuals} we therefore analyse the process of to-be-reported \emph{components} (of individuals). More precisely, a new infection/birth is immediately decided if the involved individuals would report the other (with probability $p$) or not. If it will, then the new individual belongs to the same component but if it will not, the newly infected/born will create a new to-be-reported component. The reason for studying this more complicated description of the same process is that the to-be-reported components of the limit process behave completely independent thus making it a branching process. It is hence possible to use theory for branching process to determine if the process is sub- or super-critical and derive the probability for extinction/minor outbreak. We are now ready for our first main result.

%%%%%%%%%%%%% theorem 1 %%%%%%%%%%%%%
\begin{mythm}
\label{theorem:early_approx}
For any finite time interval $[0,t_{0}],$ the SIR-TT epidemic process $E_{n}(\beta,\gamma,\delta,p)$ converges in distribution to the limit process $E(\beta,\gamma,\delta,p)$ as $n\rightarrow \infty.$
\end{mythm}
%%%%%%%%%%%%%%%%%%%%%%%%%%%%%%%%%%%%%%%%%%%%%%%%%%%%%%%%%%%%%%%%%%%%%%%%%%%%%%%%%%%%%%%%%%%%%%%%%%%%%%%%
%%%%%%%%%%%%% how the branching process is characterized as the B-D-process of components %%%%%%%%%%%%%
%%%%%%%%%%%%%%%%%%%%%%%%%%%%%%%%%%%%%%%%%%%%%%%%%%%%%%%%%%%%%%%%%%%%%%%%%%%%%%%%%%%%%%%%%%%%%%%%%%%%%%%%
When proving Theorem 
1 (Section \mbox{\ref{sec:earlystage}}) we use coupling methods (\citep{andersson_stochastic_2000,ball_strong_1995}) to show that during any finite time period, the epidemic described in terms of to-be-reported components converges to the branching process of to-be-reported components. Having done this it remains to derive properties of such a limiting to-be-reported component. It turns out that such a to-be-reported component can be described by a jump Markov chain having births (increased by 1), deaths (decreased by one) and killing (the whole component being removed), all occurring with linear rates. Suppose that there are currently $k$ alive individuals in the component, then each such individual gives birth to a new to-be-reported individual at rate $\beta p$ and thus the total birth rate is $k\beta p$. Each individual dies naturally at rate $\gamma$ so the overall death rate equals $k\gamma$. Finally, the whole component is removed as soon as one of the $k$ alive individuals is removed, so this happens at rate $k \delta.$ Until the component is removed, it generates index cases to new independent to-be-reported components at rate $k\beta (1-p).$ This describes the evolution of the to-be-reported components.
%%%%%%%%%%%%%%%%%%%%%%%%%%%%%%%%%%%%%%%%%%%%%%%%%%%%%%%%%%%%%%%%%
%%%%%%%%%%%%% effective component reproduction number %%%%%%%%%%%%%%%%%%%%%%%%%%
%%%%%%%%%%%%%%%%%%%%%%%%%%%%%%%%%%%%%%%%%%%%%%%%%%%%%%%%%%%%%%%%%
Viewed as a branching process, the most interesting quantity is the distribution of the number of offspring $Z$ (= roots of new to-be-reported components) that one to-be-reported component produces before being removed. The mean offspring distribution, corresponding to the reproduction number of the components in the epidemic setting, is then given by
\begin{equation}
\label{eq:defRc}
  R^{(c)}_{*}=\mathbb{E}[Z].
\end{equation}
By considering the jump Markov chain we can write the total offspring $Z$ as a sum
\begin{equation}
\label{eq:expressionZ}
    Z=\sum_{i=1}^{N_{C}}X_{i},
\end{equation}
where $N_{C}$ denotes the number of jumps the Markov process makes until it is removed, and $X_{i}$ denotes the number of newly generated roots of components between the $(i-1)$-th and $i$-th jump. Because all three jumps the process can make (birth, death and removal) happen at linear rates, the current number of alive individuals only affects the speed of the process but not which jump it makes. As a direct consequence, the components $X_1, X_2, \dots$ are not only independent but also identically distributed: $X_i\sim X$. It hence follows that
\begin{equation}
\label{eq:expressionRc}
    R^{(c)}_{*}=\mathbb{E}[N_{C}]\mathbb{E}[X].
\end{equation}
In Section~\ref{sec:branchingprocess} we show that
\begin{equation}
\label{eq:expressionEX}
    \mathbb{E}[X]=\frac{\beta(1-p)}{\beta p+\gamma+\delta},
\end{equation}
and
\begin{equation}
\label{eq:expressionENc}
        \mathbb{E}[N_{C}] =1+\sum_{k=1}^{\infty}\mathbf{P}(N_{c} > k),
\end{equation}
where
\begin{equation}
\label{eq:expressionPNc}
\begin{split}
     \mathbf{P}(N_{C} > k)&= \bigg(1-\sum_{j=1}^{\left \lceil{k/2}\right \rceil }\frac{1}{2j-1}\binom{2j-1}{j}{\bigg(\frac{\beta p}{\gamma+\beta p}\bigg)}^{j-1}{\bigg(\frac{\gamma}{\gamma +\beta p}\bigg)}^{j} \bigg)\\
 & \cdot \bigg(\frac{\beta p+\gamma}{\beta p+\gamma+\delta}\bigg)^{k}.
\end{split}
\end{equation}
%%%%%%%%%%%%%%%%%%%%%%%%%%%%%%%%%%%%%%%%%%%%%%%%%%%%%%%%%%%%%%%%%%%%%%%%%%%%%%
%%%%%%%%%%%%% effective individual reproduction number %%%%%%%%%%%%%%%%%%%%%%%%%%
%%%%%%%%%%%%%%%%%%%%%%%%%%%%%%%%%%%%%%%%%%%%%%%%%%%%%%%%%%%%%%%%%%%%%%%%%%%%%%
The reproduction number defined above was for the to-be-reported components (the average number of new components it produces before being removed, i.e.\ completely diagnosed or die out undetected). Even though the original limit process is not a branching process, it is possible to determine the effective reproduction number $R^{(ind)}_{*}$ for it. Its interpretation is easier: it equals the average number of individuals a typical infected infects during the early stage of the epidemic.

In Section~\ref{sec:branchingprocess} we derive the following relation between the two reproduction numbers.
\begin{equation}
\label{eq:expressionRind}
    R^{(ind)}_{*}= \frac{\mu_c-1 + R^{(c)}_{*}}{\mu_{c}} = 1-\frac{1}{\mu_c}+\frac{R^{(c)}_{*}}{\mu_{c}},
\end{equation}
with 
\begin{equation}
\label{eq:expression_mu_c}
    \mu_{c}=1+\frac{\beta p}{\beta p+ \gamma} \mathbb{E}[N_{c}]
\end{equation}
the expected number of born individuals before the component is removed. 
%%%%%%%%%%%%%%%%%%%%%%%%%%%%%%%%%%%%%%%%%%%%%%%%%%%%%%%%%%%%%%%%%%%%%%%%%%%%%%
%%%%%%%%%%%%% remark about either both R^(c) or R^(ind) </>/=1 %%%%%%%%%%%%%
%%%%%%%%%%%%%%%%%%%%%%%%%%%%%%%%%%%%%%%%%%%%%%%%%%%%%%%%%%%%%%%%%%%%%%%%%%%%%%

\begin{myrem}
It is easily observed that $R^{(ind)}_{*}<1$ if and only if $R^{(c)}_{*} < 1$, and similarly for ''='' and ''$>$''. The limit process is hence sub-critical (i.e.\ will die out with probability 1), when $R^{(c)}_{*} < 1$ and super-critical if $R^{(c)}_{*} > 1$ (so will grow beyond all limits with positive probability). The same holds true if $R^{(c)}_{*}$ is replaced by $R^{(ind)}_{*}$. This indicates the following corollary. 
\end{myrem}

\begin{myrem}
In Section \ref{sec:numerical} the two reproduction numbers are computed numerically for different parameter values. Surprisingly, the component reproduction number $R^{(c)}_{*}$ turns out not to be monotonically decreasing in the tracing probability $p$. However, the individual reproduction number seems to be decaying in $p$ as expected. We have failed in producing a formal proof of this result.
\end{myrem}
%%%%%%%%%%%%%%%%%%%%%%%%%%%%%%%%%%%%%%%%%%%%%%%%%%%%%%%%%%%%%%%%%%%%%%%%%%%%%%%%%%%%
%%%%%%%%%%%%%% corollary about minor/major outbreak %%%%%%%%%%%%%%%%%%%%%%%%%%%%%%%%%%%%%%%%%
%%%%%%%%%%%%%%%%%%%%%%%%%%%%%%%%%%%%%%%%%%%%%%%%%%%%%%%%%%%%%%%%%%%%%%%%%%%%%%%%%%%%
%% Dongni: corollary for R^ind not R^c
\begin{mycol}
\label{corollary:majoroutbreak}
Let $Z_n$ denote the final number, and $\bar Z_n=Z_n/n$ the final fraction, that get infected during the entire epidemic. If $R^{(ind)}_{*} \le 1$ it then follows that $\bar Z_n \overset{p}{\to} 0$ namely there will be a minor outbreak for sure. If $R^{(ind)}_{*}> 1,$ then $Z_n\to\infty$ with probability $1-\pi$ where 
$\pi$ is the smallest solution on $[0,1]$ of the equation
\begin{equation}
\label{eq:generatingfct}
    s=\rho_{Z}(s)=\rho_{N_{c}}(\rho_{X}(s)),
\end{equation}
where $\rho_{Z},$ $\rho_{N_{c}}$ and $\rho_{X}$ 
are the probability generating functions of $ Z, N_{C} $ and $X$, respectively.
\end{mycol}

%%%%%%%%%%%%%%%%%%%%%%%%%%%%%%%%%%%%%%%%%%%%%%%%%%%%%%%%%%%%%%%%%%%%%%%%%%%%%%%%%%%%%%%%%%%%%%%%
%%%%%%%%%%%%%%%%%%%%%%%% Special case: the SI-TT model %%%%%%%%%%%%%%%%%%%%%%%%%%%%%%%%%%%%%%%%%%%%%%%%%%%%%%%%%
%%%%%%%%%%%%%%%%%%%%%%%%%%%%%%%%%%%%%%%%%%%%%%%%%%%%%%%%%%%%%%%%%%%%%%%%%%%%%%%%%%%%%%%%%%%%%%%%

In the last part of Section~\ref{sec:branchingprocess}, we give special attention to the case where there is no natural recovery ($\gamma=0$) which accordingly can be called the \textbf{SI-TT} model. In this situation, the expressions become simpler and are given in the following corollary.

%%%%%%%%%%%%%%%%%%%%%% %%%%%%%%%%%%%%%%%%%%%% %%%%%%%%%%%%%%%%%%%%%% 
%%%%%%%%%%%%%%%%%%%%%% corollary in SI-TT model %%%%%%%%%%%%%%%%%%%%%% 
%%%%%%%%%%%%%%%%%%%%%% %%%%%%%%%%%%%%%%%%%%%% %%%%%%%%%%%%%%%%%%%%%% 
\begin{mycol}
\label{corollary:SI-TT}
In the SI-TT model having $\gamma=0$, the component reproduction number is given by \begin{equation}
\label{eq:Rc_SID}
    R^{(c)}_{*,SI-TT}=\frac{\beta(1-p)}{\delta},
\end{equation}
the individual reproduction number equals \begin{equation}
\label{eq:Rind_SID}
    R^{(ind)}_{*,SI-TT}=\frac{\beta}{\beta p+\delta},
\end{equation}
and the minor outbreak probability becomes
\begin{equation}
\label{eq:pi_SID}
    \pi=\frac{1}{R^{(c)}_{*,SI-TT}}=\frac{\delta}{\beta(1-p)}.
\end{equation}
\end{mycol}

\begin{myrem}
Again in this case, we see from Equation (\ref{eq:Rc_SID}) and (\ref{eq:Rind_SID}) that $ R^{(c)}_{*,SI-TT}$ is smaller than or equal to or larger than 1, if and only if $ R^{(ind)}_{*,SI-TT}$  is smaller than or equal to or larger than 1, respectively.
\end{myrem}
%%%%%%%%%%%%%%%%%%%%%%%%%%%%%%%%%%%%%%%%%%%%%%%%%%%%%%%%%%%%%%%%%%%%%%%%%%%%%%%%%%%%%%%%%%%%%%%%
%%%%%%%%%%%%%%%%%%%%%%% main phase of epidemic %%%%%%%%%%%%%%%%%%%%%%%%%%%%%%%%%%%%%%%%%%%%%%%%

%%%%%%%%%%%%%%%%%%%%%%%  notations %%%%%%%%%%%%%%%%%%%%%%% 
%%%%%%%%%%%%%%%%%%%%%% %%%%%%%%%%%%%%%%%%%%%% %%%%%%%%%%%%%%%%%%%%%% 
We now switch attention to the main phase of the epidemic rather than its beginning (the corresponding proofs are given in Section~\ref{sec:mainphase}). In order to surpass the initial phase of the epidemic we therefore assume a small initial \emph{fraction} $\varepsilon >0$ of infectives (instead of only one initial infective). Further, we assume that contact tracing only takes place for the contacts 
resulting in infection, not for the contacts between infectious individuals and individuals who have been infected.

We start by introducing notations for the epidemic and its limiting process, where we keep track of the fraction of susceptibles as well as the fractions of infectives belonging to to-be-reported components with each given number of infectives.

For $t\geq 0$ and $\varepsilon>0,$ let $S^{(n)}(t)$ denote the number of susceptible individuals with initial value $S^{(n)}(0)=(1-\varepsilon)n.$ For $j=1,2,...,n,$ let $I^{(n)}_j(t)$ be the number of infectious individuals that belong to a to-be-reported component containing $j$ infectives, and
$I^{(n)}(t)=\sum_{j=1}^{n}I^{(n)}_j(t)$
denotes the total number of infectious individuals at time $t$,
with initial values
$I^{(n)}(0)=I^{(n)}_1(0)=\varepsilon n,$
and
$I^{(n)}_2(0) =I^{(n)}_3 (0)=\cdots= 0.$
Let $R^{(n)}(t)$ 
denote the number of individuals who stop being infectious including both naturally recovered and diagnosed,
with initial value
$R^{(n)}(0)=0.$
Since it always hold that $S^{(n)}(t)+I^{(n)}(t)+R^{(n)}(t)=n,$ we eliminate $R^{(n)}$ from our analysis. Further, let $E^{(n)}=\{E^{(n)}(t); t \geq 0\}=\{(S^{(n)}(t)/n, I^{(n)}_1(t)/n, I^{(n)}_2(t)/n,...,I^{(n)}_n(t)/n\}$ be the stochastic epidemic density process which becomes infinite-dimensional, as the population size $n$ goes to infinity.

The limiting deterministic process denoted by $E^{\infty}= \{E^{\infty}(t); t\ge 0\}=\{s(t),i_{1}(t),i_{2}(t),\cdots\}$ is obtained by considering the jumps that the components make. An infection in a $j$-component moves the component to a $(j+1)$-component implying that $S$ is reduced by 1, $I_j$ reduced by $j$ and $I_{j+1}$ increased by $j+1$. A natural recovery in such a component increases $R$ by 1, decreases $I_j$ by $j$ and increases $I_{j-1}$ by $j-1$. Finally, a test-detection in such a component reduces $I_j$ by $j$ and increases $R$ by $j$.

In Section~\ref{sec:mainphase}, we prove of the following theorem.
 % Define the limiting system without truncation and formulate the proposition without truncation.

\begin{mythm}

\label{theorem:mainphase}
For $t \geq 0,$ let $s(t)$ be the community fraction of susceptibles, $i_j(t)$ be
the fraction of infectious individuals belonging to a to-be-reported component containing $j$ infectives, and $i(t)=\sum_{i=1}^{\infty}i_{j}(t),$ be the community fraction of infectives.

Further, we set   
\begin{equation}
\label{eq:diff_s}
    s'(t)=-\beta s(t)i(t),
\end{equation}
\begin{equation}
\label{eq:diff_i1}
{i_1}'(t)=\beta(1-p)i(t)s(t)+\gamma i_{2}(t)-\beta p i_{1}(t)s(t)-(\gamma +\delta)i_{1}(t),
\end{equation}
for $j\geq 2,$
\begin{equation}
\label{eq:diff_ij}
{i_j}'(t)= \beta {p}j i_{j-1}(t)s(t)+\gamma j i_{j+1}(t)-\beta p j i_{j}(t)s(t)-(\gamma+\delta) j i_{j}(t),
\end{equation}
%\begin{equation}
%r'(t)= \gamma i(t)+\delta \sum_{j=1}^{\infty} j i_{j}(t),
%\end{equation}
with the corresponding initial configuration
\begin{equation}
\label{eq:initial_s}
    s(0)=1-\varepsilon,
\end{equation}
\begin{equation}
\label{eq:initial_i1}
     i(0)=i_1(0)=\varepsilon,
\end{equation}
 for $j\geq 2,$
\begin{equation}
\label{eq:initial_ij}
    i_j(0)=0.
\end{equation}
%and
%\begin{equation}
    %r(0)=0.
%\end{equation}

Then the infinite-dimensional stochastic epidemic process $E^{(n)}$ converges to the deterministic process $E^{\infty}$ defined by Equation (\ref{eq:diff_s})-(\ref{eq:initial_ij}) as $n\rightarrow\infty$, on any finite time interval $[0, t_{end}]$.
\end{mythm}

%%%%%%%%%%%%%%%%%%%%%%%%%%%%%%%%%%%%%%%%%%%%%%%%%%%%%%%%%%%%%%%%%%%%%%%%%%%%%%%%%%%%%% 
% Mention after the proposition how the proof goes %%%%%%%%%%%%%%%%%%%%%%%%%%%%%%%%%%%%%%%%%%%%%%%%%%%%%%%%%%%%%%%%%%%%%%%%%%%%%%%%%%%%%% 
When proving the theorem above, we first truncate both systems such that there is a maximal component size $K$ making the processes finite dimensional, for which theory of population processes gives convergence. Then we argue that the component sizes for the original processes will be exponentially small in maximal component size, thus making the truncated models good approximations of the original processes.

If the SIR-TT model was started with one initial infective the time it takes until a fraction $\varepsilon$, i.e.\ a number $n \varepsilon$,  have been infected, tends to infinity. For this reason this initial condition does not converge to the deterministic process above. Similarly, the end of the epidemic where the final small fraction $\epsilon$ gets infected also takes longer and longer time the larger $n$ is. However, like in many similar but simpler epidemic models we expect that, when it comes to the final number getting infected, the start of the epidemic determines if there is a major outbreak or not, and end of the epidemic has negligible effect. We formulate this more precisely in the following conjecture.

\begin{mycon}\label{conjecture}
Consider the SIR-TT epidemic starting with one initially infective. The final fraction infected $\bar Z_n$ converges to a two point distribution $\zeta,$ where $\zeta=0$ (minor outbreak) happens with probability $\pi$, and with probability $1-\pi$, $\zeta = r_\infty= \lim_{\varepsilon\to 0}\lim_{t\to \infty} r(t)$ (major outbreak), where $\pi$ is defined in Corollary \ref{corollary:majoroutbreak} and $r(t)=1-s(t)-i(t)$ in Theorem \ref{theorem:mainphase}.
\end{mycon}

\begin{myrem}\label{remark_CLT}
 In Section~\ref{sec:numerical}, we show several simulations in support of Conjecture \ref{conjecture} and also indicating that the distribution of $\bar Z_n$ appears to satisfy a central limit theorem concentrated around the deterministic limit $r_\infty$.  
\end{myrem}

Finally in Section~\ref{sec:numerical} we perform simulations and numerical illustrations confirming our results and investigating the effect of Test-and-Trace strategy for parameter values inspired from the Covid-19 pandemic.

%%%%%%%%%%%%%%%%%%%%%%%%%%%%%%%%%%%%%%%%%%%%%%%%%%%%%%%%%%%%%%%%%%%%%%%%%%%%%%%%%%%%%%%%%%%%%%%%
%%%%%%%%%%%%%%%%%%%%%%%% Section 2 Early Stage Approximation of the Epidemic %%%%%%%%%%%%%%%%%%%
%%%%%%%%%%%%%%%%%%%%%%%%%%%%%%%%%%%%%%%%%%%%%%%%%%%%%%%%%%%%%%%%%%%%%%%%%%%%%%%%%%%%%%%%%%%%%%%%%
\section{Proof of Theorem 1}
\label{sec:earlystage}

In this section, we aim to approximate the early stages of the epidemic using large population approximations.
We first denote the sequence of our epidemic processes with one initial infective by $\{E_{n}(\beta,\gamma,\delta,p),n\geq 1\},$ where we recall that $\beta$ is rate of infection, the rate of natural recovery is $\gamma$, $\delta$ denotes the testing (diagnosis) rate and the probability of a contact being reported equals $p$. 

Then we describe the limiting process denoted by $E(\beta,\gamma,\delta,p).$ At time $t=0,$ there is only one initial ancestor. Each individual gives birth at rate $\beta$ during their lifetimes, dies naturally (naturally recovered) with rate $\gamma$ and is removed (diagnosed) with rate $\delta.$ Once removed, each of its descendants and its parent is said to be reported and immediately removed independently with probability $p.$ Meanwhile, every parent and offspring of those who are removed, will be removed independently with probability $p$ as well and so on. In particular, if a to-be-reported individual has been already died naturally, its alive to-be-reported offspring (or parent) will also be removed. Moreover, we notice that each not to-be-reported offspring becomes a new ancestor which independently produces a process in the same pattern.
Finally we show the proof of Theorem~\ref{theorem:early_approx}. 
%At the beginning of the epidemic, and assuming that the community size $n$ is large, it is very unlikely that an infected individual has contact with another infective. This  suggests that the process of infected individuals during the early phase of the outbreak, might be approximated by some kind of branching process. 
\begin{proof}[Proof of Theorem~{\upshape\ref{theorem:early_approx}}]
%%%%%%%%%%%%%%%%%%%%%%%%%%%%%%%%%%%%%%%%%%%%%%%%%%%%%%
% describe the limiting process in terms of individuals: 
%%%%%%%%%%%%%%%%%%%%%%%%%%%%%%%%%%%%%%%%%%%%%%%%%%%%%%
First of all, it is worth noting that the two processes $E_{n}(\beta,\gamma,\delta,p)$ and $E(\beta,\gamma,\delta,p)$ behave the same way besides one slight difference. An infection occurs in the epidemic whenever a birth occurs in the branching process, whereas an infection is "effective" only if an susceptible gets infected. And in the $n$-th epidemic, the probability that an infective gives new infections to susceptibles is $S_{n}(t)/(n-1)$ $ (\approx S_{n}(t)/n$ when n large), where $S_{n}(t)$ is the number of susceptibles at time $t.$ So, an infective infects new individuals at rate ${\beta S_{n}(t)}/{n}.$ In contrast to that, an alive individual in the limiting process give birth at rate $\beta.$ However, if the size of population $n$ is large and in the beginning of epidemic we have $S_n(t)\approx n$, then we have 
${\beta S_{n}(t)}/{n} \approx \beta,$ i.e. the rate of new infection to susceptibles in $E_{n}(\beta,\gamma,\delta,p)$ is close to the birth rate in $E(\beta,\gamma,\delta,p).$ 

%%%%%%%%%%%%%%%%%%%%%%%%%%%%%%%%%%%%%%%%%%%%%%%%%%%%%%
% describe the component branching process: 
%%%%%%%%%%%%%%%%%%%%%%%%%%%%%%%%%%%%%%%%%%%%%%%%%%%%%%
As compared to the early stage approximation of standard SIR epidemic \citep{andersson_stochastic_2000,ball_strong_1995}, the number of alive individuals in this limiting process $E(\beta,\gamma,\delta,p)$ behaves not like a branching process, since it is possible that several death occur at the same time and thus the jumps of this limiting process can not only be up or down by one. 

On the other hand, if the limiting process is described in terms of to-be-reported components, then it behaves like a branching process. A component starts with one newly born (infected) individual which would not report its infector and we call this individual the root of this component. During its life duration (infectious period) this individual gives birth to new individuals, some of which will be reported and others will not. Each of those not-to-be-reported individuals becomes a root of new components, whereas those who will be reported belong to the same component. Given that there are currently $k$ to-be-reported individuals in one component, then each such individual gives birth at rate $\beta$, where each newborn belongs to the same component with probability $p$ and generates a new component with probability $1-p$. Thus, each of individual in this component gives birth to new to-be-reported individuals at rate $\beta p$ and thus the total birth rate is $k\beta p.$  Each alive (infectious) individual dies naturally(naturally recovered) at rate $\gamma$. The whole component is diagnosed if and only if one of those $k$ to-be-reported individuals is diagnosed, implying that the death rate of this process of components is $k \delta.$ Until all these $k$ individuals are removed, it generates roots of new components at rate $k\beta (1-p).$ This describes the birth and death of the to-be-reported components.

By applying the coupling method (see \citep{andersson_stochastic_2000, ball_strong_1995}), we show that the epidemic process $E_{n}(\beta,\gamma,\delta,p)$ described in terms of components converges in to the branching process of to-be-reported components. A contact (infection) in the epidemic corresponds to a birth in the branching process. Obviously, the branching process and the epidemic process of components are perfectly coupled with each other up until the time $T_n$ when the first "ghost" appears, where by "ghost" we mean the newly contacted individual which has been infected in the epidemic. If we label each $i-$th contact as $c_i$, then for any time $t_{0} \geq 0,$ the event $T_n \geq t_0$ that there has been no "ghost" occur before time $t_{0}$, is equivalent to the case that all the contacts $c_1, ..., c_{[t_{0}]}$ are distinct. Using the classic birthday-problem method, we see that 
\begin{equation*}
    \mathbf{P}(T_n \geq t_0) \approx e^{-n^2/2[t_{0}]} \rightarrow 1,
\end{equation*}
as $n\rightarrow \infty.$
This completes the proof of Theorem \ref{theorem:early_approx}. 
\end{proof}
\section{Properties of the limiting branching process }
\label{sec:branchingprocess}
Now we explore the properties of the limiting branching process $E(\beta,\gamma,\delta,p)$ of to-be-reported components which can be used to approximate the epidemic during the early phase. 
\subsection{Process of the to-be-reported components in the full model}

First, we note that our reporting process can be decided in advance, and recall that we use the same reporting or not decision in both ways between each pair of individuals since at most one direction will be used. We then focus on this Markov jump process of the components having births, deaths and sudden killing of the whole component. 

We define the size $k$ of an to-be-reported component by the number of alive (infectious) individuals in the component and hence ignore the dead (recovered) individuals. A component currently having size $k$ produces new roots of new components at rate $k\beta(1-p).$ The component itself remains with size $k$ for an exponentially distributed time with rate $k (\beta p+\gamma+\delta)$, next event would be one of the three following independent cases. The first case is that a new infection occurs at rate $k \beta p$, which corresponds to increasing the size of component by one. Secondly, we note that each of the individuals in the component becomes naturally recovered at rate $\gamma$. If this happens, then the size of component would decrease by one. The remaining case is that one of the to-be-reported individuals is diagnosed and so the whole component is eliminated by diagnosis at rate $k\delta$, which means that upon this event, the size of component goes down to zero.

\begin{figure}
\centering
  \includegraphics[width=119mm]{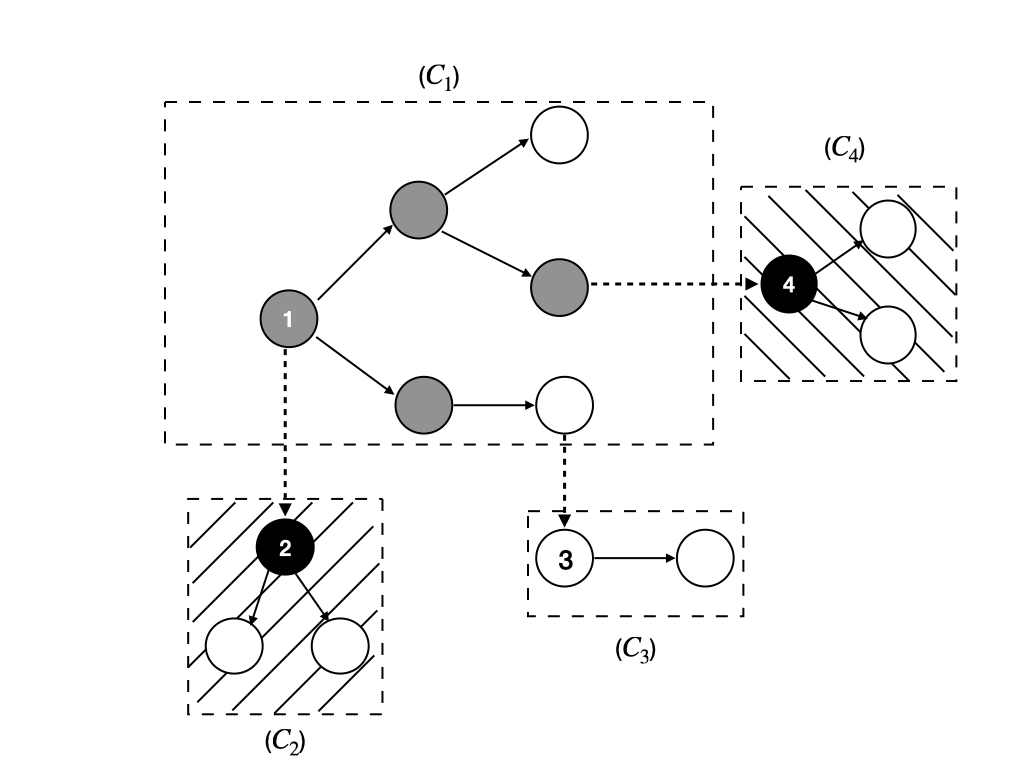}
\caption{Example of a "reporting tree": The white nodes stands for "infectious", the grey ones for "naturally recovered", whereas the black ones for "diagnosed". A directed edge from one node to another means that the latter one is infected by the previous one. Full edges reflect to-be-reported contacts (probability $p$) and the dashed ones for those not to be reported}
\label{report-tree}
\end{figure}

In Fig.~\ref{report-tree}, we show an example illustrating how a \textbf{"reporting branching tree"} of to-be-reported components grows: at first, we have a newly infected case, namely the node 1. An edge goes from one node to another node, meaning that the latter one is infected by the previous one. The dashed edge between two nodes means for the not-to-be-reported case, whereas the full edge stands for the to-be-reported case. After a certain period, there is a to-be-reported component, denoted by $C_{1},$ produced by its root  1, and three newly generated roots 2, 3 and 4, each of which produces their own to-be-reported components, denoted by $C_{2}$, $C_{3}$ and $C_{4}$ respectively. Furthermore, the white nodes stands for "infectious", the grey ones for "naturally recovered" and the black ones for "diagnosed". We can see that at this stage when the roots 2 and 4 are diagnosed, the whole components $C_{2}$ and $C_{4}$ are reported and immediately diagnosed.

Our interest is to derive the important quantity for our epidemic model, namely \textbf{the effective component reproduction number} $R^{(c)}_{*}$, which is defined as the expected number $R^{(c)}_{*}=\mathbb{E}[Z]$ of roots of new components generated by one root before its component is removed (completely diagnosed or dies out undetected).  Since we aim at examining the effect of testing and tracing, so given fixed rates $\beta$ and $\gamma,$ we consider $R^{(c)}_{*}=R^{(c)}_{*}(\delta,p)$ as a function of testing rate $\delta$ and tracing probability $p.$ Later in Section~\ref{sec:numerical}, we will show how the $R^{(c)}_{*}$ varies with the testing fraction $\delta/(\delta+\gamma)$ and tracing probability $p.$

To find the distribution of $Z$, we first discuss the number of events before the whole component is removed by computing the probability that the whole component is not removed before $k$ events. We recall that at each time of event, there would be only one of three following events occurs. A birth occurs with rate $k \beta p$, whereas the death of whole component happens with rate $k\delta$ and the size of component decreases by one with rate $k\gamma$. As a consequence the probability of giving a birth, which corresponds to increasing the component size by one, is given by 
\begin{equation*}
    \frac{k\beta p}{k\beta p+k\gamma+k\delta}=\frac{\beta p}{\beta p+\gamma+\delta},
\end{equation*}
the probability of death of the whole component equals
\begin{equation*}
    \frac{k\delta}{k\beta p+k\gamma+k\delta}=\frac{\delta}{\beta p+\gamma+\delta},
\end{equation*}
and the probability that the component size decreases by one, is given by
\begin{equation*}
    \frac{k\gamma}{k\beta p+k\gamma+k\delta}=\frac{\gamma}{\beta p+\gamma+\delta}.
\end{equation*}

A different way of describing the evolution of the component is to consider increases and decreases by one (a simple random walk!) until some time when the whole component dies simultaneously. It is worth pointing out that the random walk may reach zero by itself and hence stop before a simultaneous death. The non-symmetric simple random walk $\{S_{m}, m\geq 0\}$ on $\mathbb{Z}$ starts at 1 ($S_{0}=1$)
and for ${m}=1,2,3,\cdots,$ each jump of the random walk is independent and identically distributed with 
the jump probabilities 
\begin{equation*}
    \pi_{rw}=\mathbf{P}(S_{m}-S_{m-1}=1)=\frac{\beta p}{\gamma+\beta p},
\end{equation*}
and 
\begin{equation*}
    \mathbf{P}(S_{m}-S_{m-1}=-1)=1-\pi_{rw}=\frac{\gamma}{\gamma+\beta p}.
\end{equation*}
On top of this, each jump may result in diagnosis of the whole component (simultaneous death), and each time at which this happens with probability $\delta/(\beta p + \gamma + \delta)$. The number $N_{D}$ of events until the whole component is eliminated by diagnosis is hence geometrically distributed with parameter $\delta/(\beta p + \gamma + \delta)$.

Next we derive the probability that the random walk does not hit zero before $k$ jumps. Let 
\begin{equation*}
    N_{rw}=\inf_{m\geq 0}\{S_{m}=0\}
\end{equation*}
denote the first hitting zero time of random walk and it is clear that only odd steps can be taken in order to hit the origin. So, for $m$ even, the probability 
we have 
    $\mathbf{P}(N_{rw}=m)=0$. 
Otherwise $m=2j-1$ for $j=1,2,\cdots,$ we apply the Hitting Time Theorem in \citep{hofstad_elementary_2008}, which yields that the probability of first hitting zero at $m-$th step is given by 
\begin{equation*}
    \mathbf{P}(N_{rw}=m)=\mathbf{P}(N_{rw}=2j-1)=\frac{1}{2j-1}\mathbf{P}(S_{2j-1}=0),
\end{equation*}
where the probability of the (unrestricted) random walk equals 0 at $m-$th step is
\begin{equation*}
    \mathbf{P}(S_{2j-1}=0)=\binom{2j-1}{j}{\pi_{rw}}^{j-1}{(1-\pi_{rw})}^{j},
\end{equation*}
since in this case, the random walk must have taken $(j-1)$ up-jumps and $j$ down-steps.
We conclude that the probability of the random walk \emph{not} hitting zero before $k$ steps equals
\begin{equation*}
\begin{split}
    \mathbf{P}(N_{rw}>k)&=1-\sum_{m=1}^{k}\mathbf{P}_1(N_{rw}=m)\\
    &=1-\sum_{j=1}^{\left \lceil{k/2}\right \rceil }\mathbf{P}_1(N_{rw}=2j-1)\\
    &=1-\sum_{j=1}^{\left \lceil{k/2}\right \rceil }\frac{1}{2j-1}\binom{2j-1}{j}{\pi_{rw}}^{j-1}{(1-\pi_{rw})}^{j}.
\end{split}
\end{equation*}
Moreover, let $N_{C}$ denote the number of jumps up until the whole to-be-reported component is extinct (either from simultaneous diagnosis or all individuals having recovered naturally), i.e. 
\begin{equation*}
  N_{C}=\min\{N_{rw},N_{D}\}.
\end{equation*}
Recalling that $N_D$ is geometric distributed, the probability that the whole component has not gone extinct before $k=1, 2, \cdots ,$ events is given by
\begin{align*}
\mathbf{P}(N_{C} > k) & = \mathbf{P}(N_{rw}>k)\cdot \mathbf{P}(N_{D}>k)
\\
& = 
\bigg(1-\sum_{j=1}^{\left \lceil{k/2}\right \rceil }\frac{1}{2j-1}\binom{2j-1}{j}{\bigg(\frac{\beta p}{\gamma+\beta p}\bigg)}^{j-1}{\bigg(\frac{\gamma}{\gamma+\beta p}\bigg)}^{j} \bigg) \bigg(\frac{\beta p+\gamma}{\beta p+\gamma+\delta}\bigg)^{k}.
\end{align*}

Now, it is sufficient to analyze the number $X_{i}$ of newly generated roots of components between each $(i-1)-$th and $i-$th jump. Given a to-be-reported component of size $k$ at that time, the roots of new components are generated at rate $k\beta(1-p)$ for an exponential time of parameter $k (\beta p+\gamma+\delta).$ This implies that the distribution of $X_{i}$ is geometrically distributed with 
parameter
\begin{equation*}
    1-\frac{k\beta(1-p)}{k(\beta p+\gamma+\delta)+k\beta(1-p)}
    =\frac{\beta p+\gamma+\delta}{\beta +\gamma+\delta}.
\end{equation*}
This parameter is independent of $k$ implying that the variables $X_{1},X_{2},\dots $, are identically and independently geometrically distributed as $X$.
So, between any two jumps, the probability of $k$ newly generated roots is given by
\begin{equation*}
    \mathbf{P}(X = k)=\left( \frac{\beta(1-p)}{\beta +\gamma+\delta} \right)^{k} \frac{\beta p+\gamma+\delta}{\beta +\gamma+\delta}.
\end{equation*}

Based on the former discussion, we conclude that the total number of roots of new components produced by a to-be-reported component can be written as
 \begin{equation*}
     Z=\sum_{i=1}^{N_{C}}X_{i}.
 \end{equation*}
As stated in Equation (\ref{eq:expressionRc}), due to independence, it follows  that 
\begin{equation*}
    R^{(c)}_{*}=\mathbb{E}[Z]=\mathbb{E}[N_{C}]\cdot \mathbb{E}[X],
\end{equation*}
where the expectation of $X$ is given by
\begin{equation*}
    \mathbb{E}[X]= \frac{\beta(1-p)}{\beta p+\gamma+\delta},
\end{equation*}
and for $\mathbb{E}[N_{C}]$ we have
\begin{equation*}
        \mathbb{E}[N_{C}] =1+\sum_{k=1}^{\infty}\mathbf{P}(N_{c} > k).
\end{equation*}

\subsection{Proof of Corollary 1}

In the following text we give the proof of Corollary \ref{corollary:majoroutbreak}. 

\begin{proof}[Proof of Corollary~{\upshape\ref{corollary:majoroutbreak}}]
Intuitively, we note that the limiting process of components will become extinct when $\mathbb{E}[Z] \leq 1.$ This implies that a minor outbreak will occur, if the component reproduction number $R^{(c)}_{*}=\mathbb{E}[Z]$ is smaller than or equal to one. 
Next, regarding to the situation when $R^{(c)}_{*}>1,$ the branching process is possible to explode, and so there will be a major outbreak in the epidemic. Now we put our focus on the probability $\pi$ of minor outbreak and the probability of major outbreak, namely $1-\pi$. 

We assume that there is one initial infective. As discussed in previous section, the probability of minor outbreak in the epidemic can be approximated by the probability of extinction in the limiting process at the early stage of outbreak. Given $k$ newly generated roots, the conditional probability of extinction is then clearly $\pi^{k}$. Thus, the probability $\pi$ is the solution on $[0,1]$ of the following equation:
\begin{equation}
\label{eq:minor_outbreak}
        \pi=\sum_{k=1}^{\infty}\pi^{k}\mathbf{P}(Z=k),
\end{equation}
where we note that the right-side of Equation (\ref{eq:minor_outbreak}) is exactly the probability generating function $\rho_{Z}(\pi)$ of $Z.$ For the computation of $\rho_{Z}$, we have
\begin{equation*}
%\begin{split}
    \rho_{Z}(s)=\mathbb{E}[s^{Z}]=\mathbb{E}[s^{\sum_{i=1}^{N_{C}}X_{i}}]
           =\rho_{N_{C}}( \rho_{X}(s)),
%\end{split}
\end{equation*}
where the probability generating function $\rho_{X}$ of $X$ is given by
\begin{equation*}
    \rho_{X}(s)=\mathbb{E}[s^{X}]=\frac{\theta}{1-(1-\theta)s}
\end{equation*}
with $\theta = ({\beta p+\gamma+\delta})/({\beta +\gamma+\delta}),$
and the probability generating function $\rho_{N_{C}}$ of $N_{C}$ is given by
\begin{equation*}
    \rho_{N_{C}}(t)=\sum_{k=1}^{\infty}t^{k}\mathbf{P}(N_{C}=k)
\end{equation*}
with 
\begin{equation*}
    \mathbf{P}(N_{C}=k)%&=\mathbf{P}(\min\{N_{rw},N_{D}\}=k)\\
    =\mathbf{P}( N_{rw} = k)\mathbf{P}(N_D \geq k) + \mathbf{P}( N_{rw} > k)\mathbf{P}(N_D= k).
\end{equation*}
Finally we obtain the probability generating function $\rho_{Z}$ of $Z$:
\begin{equation}
\label{eq:rho_z}
    \begin{split}
       \rho_{Z}(s)&= \sum_{k=1}^{\infty}\bigg(\frac{\theta}{1-(1-\theta)s}\bigg)^{k}\mathbf{P}(N_{C}=k)\\
       &= \frac{\theta}{1-(1-\theta)s}\cdot (1-\mathbf{P}(N_{C}>1))\\
       &+\sum_{k=2}^{\infty}\bigg(\frac{\theta}{1-(1-\theta)s}\bigg)^{k}(\mathbf{P}(N_{C}>k-1)-\mathbf{P}(N_{C}>k)).
    \end{split}
\end{equation}
Solving the equation 
\begin{equation*}
    s=\rho_{Z}(s)
\end{equation*}
with Equation (\ref{eq:rho_z}) on $[0,1]$ gives us the smallest solution $\pi,$ which equals the probability of minor outbreak in the epidemic with one initial infective.
In addition to that, if there are initially $m$ infectives, small outbreak occurs with probability $\pi^{m}.$
\end{proof}
%%%%%%%%%%%%%%% the R_ind: %%%%%%%%%%%%%%%
We then aim to derive the effective \emph{individual} reproduction number $R^{(ind)}_{*}=R^{(ind)}_{*}(\delta,p)$, which equals the expected number of infected cases generated by a random infectious individual (before being tested or recovering). 

We start with a new expression for our effective component reproduction number. Let $I_{c}$ be the overall number of individuals who have been infected in a to-be-reported component before it goes extinct, starting with one single infectious individual. For $k \geq 1,$ let
\begin{equation*}
    p_{k}=\mathbf{P}(I_{c}=k)
\end{equation*}
be the probability that there have been in total $k$ individuals infected in a component,
and let $r_{k}$ be the expected number of new roots generated by such a component, given that $I_{c}=k.$ It then follows that
\begin{equation*}
   R^{(c)}_{*} =\sum_{k=1}^{\infty} r_{k}p_{k}.
\end{equation*}
Further, during the early stage of epidemic, the probability $\Tilde{p}_{k}$ that an individual belongs to a component with $I_{c}=k,$ is given by 
\begin{equation*}
    \Tilde{p}_{k}=\frac{k p_{k}}{\mu_{c}}
\end{equation*}
(the size-biased distribution) with
\begin{equation*}
    \mu_{c}:=\mathbb{E}[I_{c}]=\sum_{j=1}^{\infty} j p_{j}.
\end{equation*}
Given that $I_{c}=k,$ there would be $(k-1)$ infections occurred in the component and $r_{k}$ average infections out of the component before it dies. In total such a component hence on average generate $k-1 + r_k$ infections and randomly chosen individual hence infects $((k-1)+r_{k})/k$ on average.
This implies that our effective individual reproduction number is given by 
\begin{equation}
\label{eq:expression_R_ind}
    R^{(ind)}_{*}=\sum_{k=1}^{\infty}\frac{(k-1)+r_{k}}{k} \Tilde{p}_{k} = \sum_{k=1}^{\infty} \left( (k-1)+r_{k}\right) \frac{p_{k}}{\mu_{c}} .
\end{equation}
Simplifying the expression for $R^{(ind)}_{*}$ in Equation (\ref{eq:expression_R_ind}) gives us
\begin{equation*}
    R^{(ind)}_{*}=1-\frac{1}{\mu_{c}}+\frac{R^{(c)}_{*}}{\mu_{c}}.
\end{equation*}
The equation above shows us that $R^{(ind)}_{*}$ is smaller than, equal to, or larger than 1, if and only if $R^{(c)}_{*}$ is smaller than, equal to, or larger than 1. 

Intuitively, we can explain this relation between $R^{(ind)}_{*}$ and $R^{(c)}_{*}$ as follows. On one hand, the $R^{(c)}_{*}$ can be considered as the average infections produced outside the component. On other hand, $\mu_{c}$ is the average number of individuals who have been infectious in the component. We notice that one of $\mu_{c}$ would be the root of this component, while the other $\mu_{c}-1$ are internal infections in this component. Then there are in total $(R^{(c)}_{*}+\mu_{c}-1)$ infections on average by this component. Hence, the average number of infections per individual becomes $(R^{(c)}_{*}+\mu_{c}-1)/\mu_{c}.$    
%%%%%%%%%%%%%%%%%%%%%%%%%%%%%%%%%
%%%%%%%%%%% compute mu_c %%%%%%%%%%%
%%%%%%%%%%%%%%%%%%%%%%%%%%%%%%%%%%%%%%%%%%%%

It remains to compute the expected number $\mu_{c}$ of the individuals who have been infectious in a component, since we have already derived $R^{(c)}_{*}$. We start by letting $J_{+}=\sum_{k=1}^{\infty} I_{k}$ denote the number of up-jumps of the random walk $\{S_{n}, n\geq 0\}$ before it dies out, where $I_{k}$ is the indicator variable with $I_{k}=1$ if the $k$-th jump of the random walk is an up-jump, i.e. for any $k\geq 1,$
\begin{equation*}
    \mathbf{P}(I_{k}=1\vert N_{c} > k-1)=\pi_{rw}
\end{equation*}
and 
\begin{equation*}
    \mathbf{P}(I_{k}=1\vert N_{c}\leq k-1)=0.
\end{equation*}
Then we conclude that 
\begin{equation*}
    \mu_{c}=1+\mathbb{E}[J_{+}]=1+\sum_{k=1}^{\infty} \mathbf{P}(I_{k}=1)=1+\pi_{rw} \sum_{k=1} \mathbf{P}(N_{c}>k-1)=1+ \frac{\beta p}{\beta p +\gamma} \mathbb{E}[N_{c}].
\end{equation*}
Together with Equation (\ref{eq:expressionRind}), we prove the final expression for $R^{(ind)}_{*}$ given by 

\begin{equation}
\label{eq:exact_R_ind}
    R^{(ind)}_{*} =\frac{R^{(c)}_{*}+\frac{\beta p}{\beta p +\gamma} \mathbb{E}[N_{c}]}{1+\frac{\beta p}{\beta p +\gamma} \mathbb{E}[N_{c}]}.
\end{equation}

\begin{myrem}\label{remark_ind}
This individual reproduction number has 
the correct threshold property, since it equals 1 exactly when 
$R_*^{(c)}$ does. However, $R_*^{(ind)}$ cannot be interpreted as the 
average number of infections caused by infected people in the beginning 
of the outbreak. This is because of delicate timing of events issues, 
closely related to those explained in \citep{ball_trapman_2016}.

\end{myrem}

\subsection{The limiting process in the {SI-{TT}} Model}
In the \textbf{SI-TT} model there is no natural recovery: $\gamma =0$. This special case turns out to give simpler explicit expressions.  The reason for the simplification is that a component can then \textit{only} go extinct due to an infectious individuals being diagnosed (resulting in the whole to-be-reported component being contact traced) whereas for the general model extinction may also happen because all individuals has recovered before a new to-be-reported infection took place.

We are interested in the number $Z_{SI-TT}$ of roots of new components produced by a to-be-reported component before it dies (i.e.\ is diagnosed). As before, the current number of infectious individuals does not affect the probability of the next jump being a new root or a diagnosis event. We can hence neglect the infections and conclude that there will be a geometrically distributed number of roots produced before the component is diagnosed. The parameter of the geometric distribution is simply the probability of a diagnosis rather than a new root: $\delta/(\delta+\beta(1-p))$. In conclusion, we have for any $k=0,1,2,\cdots$ that
%the conditional density function of $Z_{SID}$ given $Y$ as
%\begin{equation*}
    %\mathbf{P}(Z_{SID}=n|Y)=e^{-Y}\frac{Y^{n}}{n!},
%\end{equation*}
the unconditional density of $Z_{SI-TT}$ is given by
\begin{equation*}
    \begin{split}
        \mathbf{P}(Z_{SI-TT}=
        k)%&=\int_{0}^{+\infty}(e^{-y}\frac{y^{n}}{n!})\cdot(\lambda e^{-\lambda y})dy\\
        %&=\frac{\lambda}{(1+\lambda)^{n+1}}\int_{0}^{+\infty}e^{-y}\frac{y^{n}}{n!}dy\\
        %&=\frac{\lambda}{(1+\lambda)^{n+1}}\cdot \frac{\Gamma(n+1)}{n!}\\
        &=\big(1-\frac{\delta}{\delta+\beta(1-p)}\big)^{k}\frac{\delta}{\delta+\beta(1-p)}.
    \end{split}
\end{equation*}

Then in the SI-TT model, we are able to derive the \textbf{effective component reproduction number} $R^{(c)}_{*,SI-TT}$ as
\begin{equation*}
    R^{(c)}_{*,SI-TT}=\mathbf{E}[Z_{SI-TT}]=\frac{\beta(1-p)}{\delta}.
\end{equation*}
Again following the idea proving Corollary \ref{corollary:majoroutbreak}, we show Corollary \ref{corollary:SI-TT} as follows. 
\begin{proof}[Proof of Corollary~{\upshape\ref{corollary:SI-TT}}]
By finding the smallest solution $s$ on $[0,1]$ of equation
$
    s=\rho(s)
$
with $\rho(\cdot)$ the probability generating function of $Z_{SI-TT},$
we obtain that if $ R^{(c)}_{*,SI-TT} > 1,$ the probability of minor epidemic outbreak equals
\begin{equation}
    \pi=\frac{\delta}{\beta(1-p)}=\frac{1}{R^{(c)}_{*,SI-TT}}.
\end{equation}
In the case when $R^{(c)}_{*,SI-TT} \leq 1,$ the branching process will be extinct with probability $\pi=1$, implying that a major outbreak occurs with probability 0.
\end{proof}
Moreover, using the same idea of computing the \textbf{effective individual reproduction number}
in general case, here we first have the expected number of infected cases generated by the root of a component before diagnosed given by 
\begin{equation*}
    \mu_{c,SI-TT}=1+\mathbb{E}[N^{(SI-TT)}_{c}]=1+\frac{\beta p}{\delta}=\frac{\beta p+\delta}{\delta}. 
\end{equation*}
Then as stated in Section \ref{sec:main-results}, the effective individual reproduction number for this case without natural recovery has the form 
\begin{equation}
\label{eq: expression_R_TT}
%\label{eq: expression_R_TT}
    R^{(ind)}_{*,SI-TT}=\frac{\mu_{c,SI-TT}+R^{(c)}_{*,SI-TT}-1}{\mu_{c,SI-TT}}= \frac{\beta}{\beta p+\delta}.
\end{equation}

\begin{myrem}\label{remark_ind_SI}
Similar to Remark \ref{remark_ind}, this individual reproduction number possesses the correct threshold property but not the traditional interpretation as the average number of infections per individual in the beginning of the epidemic. It is also easily observed from Equation (\ref{eq: expression_R_TT}) that the effective individual reproduction number in this SI-TT model is monotone decreasing with tracing probability $p.$

\end{myrem}

%%%%%%%%%%%%%%%%%%%%%%%%%%%%%%%%%%%%%%%%%%%%%%%%%%%%%%%%%%%%%%%%%%%%%%%%%%%%%%%%%
%%%%%%%%%%%%%%%%%%%%%%%%%%% Main Phase Approximation %%%%%%%%%%%%%%%%%%%%%%%%%%%
%%%%%%%%%%%%%%%%%%%%%%%%%%%%%%%%%%%%%%%%%%%%%%%%%%%%%%%%%%%%%%%%%%%%%%%%%%%%%%%%%
\section{Proof of Theorem 2}
\label{sec:mainphase}
In previous section, we applied coupling methods to approximate the epidemic at its early stage. In this section, we give an approximation of the main phase, where the epidemic is initiated with positive fraction $\varepsilon$ of infectives. Here, we describe our original full model in the way of evolution of the clumps. By "clumps", we mean the to-be-reported components, and we only need to keep track of number of infectious individuals in each clump. %% why only need to?
 In a population of size $n$, we assume that the {number} of initial infectives equals $\varepsilon n$ and the {number} of initial susceptibles equals $(1-\varepsilon )n$. At time $t\geq 0,$ let $S^{(n)}(t)$ be the number of susceptible individuals with initial value 
\begin{equation}
    S^{(n)}(0)=(1-\varepsilon )n.
\end{equation}
For $j=1,2,...,n,$ let $I^{(n)}_j(t)$ be the number of  individuals that are infectious and belong to a to-be-reported component currently containing $j$ infectives. So, we have the total number of infectious individuals at time $t$, 
\begin{equation}
    I^{(n)}(t)=\sum_{j=1}^{n}I^{(n)}_j(t),
\end{equation}
with initial values
\begin{equation}
    I^{(n)}(0)=I^{(n)}_1(0)=\varepsilon  n,~I^{(n)}_2(0)=\cdots=I^{(n)}_n(0)=0.
\end{equation}
Let $R^{(n)}(t)$ 
denote the number of individuals which are recovered (counting both naturally recovered and diagnosed)
with initial value
\begin{equation}
    R^{(n)}(0)=0.
\end{equation}
It is then clear that for any time $t \geq 0,$
\begin{equation}
    S^{(n)}(t)+I^{(n)}(t)+R^{(n)}(t)=n.
\end{equation}

%%%%%%%%%%%%%%%%%%%%%%%%%%%%%%%%%%%%%%
Next we prove Theorem \ref{theorem:mainphase}, stating that the stochastic epidemic process denoted by
\begin{equation*}
    E^{(n)}(t)=(S^{(n)}(t)/n
,I^{(n)}_1(t)/n,I^{(n)}_2(t)/n,...,I^{(n)}_n(t)/n),
\end{equation*}
converges to a deterministic process.
\begin{proof}[Proof of Theorem~{\upshape\ref{theorem:mainphase}}]
%%%%%%%%%%%%%%%%%%%%%%%%%%%%%%%%%%%%%%%%%%%%%%%%
%%%%%%%%%%%% truncation of clump size %%%%%%%%%%
%%%%%%%%%%%%%%%%%%%%%%%%%%%%%%%%%%%%%%%%%%%%%%%%
Below we study the corresponding truncated processes by maximizing the clump sizes to some large positive integer $K$. The corresponding processes are finite dimensional for which we apply theory for density dependent population processes. These results can then be extended to the original infinite dimensional systems (with arbitrary clump size) by observing that that the maximal clump sizes are exponentially small in $K$. As a consequence, the truncated processes can be made arbitrarily close to the original infinite dimensional processes by choosing $K$ large enough. The fact that the processes are exponentially small is a direct consequence of that the epidemic processes SIR-TT may be dominated by the SI-TT (without natural recovery), and this process will have a geometrically distributed maximal clump size with parameter $\delta/(\beta p +\delta)$. We omit the details of this argument and now should show that the truncated stochastic epidemic process converges to the truncated deterministic system.

%%%%%%%%%%%%%%%%%%%%%%%%%%%%%%%%%%%%%%%%%%%%%%%%%%%%%%%%%%%%%%%%%%%%%%%%%%%%%%%%%%%%
%%%%%%%%%%%%%% truncated stochastic converges to finite-dim. system %%%%%%%%%%%%%%
%%%%%%%%%%%%%%%%%%%%%%%%%%%%%%%%%%%%%%%%%%%%%%%%%%%%%%%%%%%%%%%%%%%%%%%%%%%%%%%%%%%%
More precisely, using Kurtz's theory of Markovian Population processes \citep{andersson_stochastic_2000}, we show that the truncated stochastic "density" process, denoted by
\begin{equation*}
    E^{(n)}_{K}(t)=(S^{(n)}(t)/n,I^{(n)}_1(t)/n,I^{(n)}_2(t)/n,...,I^{(n)}_K(t)/n)
\end{equation*}
converges to a $K-$dimensional deterministic process 
\begin{equation*}
    {E^{\infty}_K}(t)=(s(t),i_1(t),i_2(t),...,i_K(t)),
\end{equation*}
which is defined by the finite system of differential equations as below.
\begin{equation}
\label{eq:diff_s(t)}
    s'(t)=-\beta s(t)i(t),
\end{equation}
\begin{equation}
\label{eq:diff_i1(t)}
{i_1}'(t)=\beta(1-p)i(t)s(t)+\gamma i_{2}(t)-\beta p i_{1}(t)s(t)-(\gamma +\delta)i_{1}(t),
\end{equation}
and for $j=2,3,...,(K-1)$ we have 
\begin{equation}
\label{eq:diff_ij(t)}
{i_j}'(t)= j\beta {p}i_{j-1}(t)s(t)+j\gamma i_{j+1}(t)-j\beta p i_{j}(t)s(t)-j(\gamma+\delta)i_{j}(t),
\end{equation}
whereas in the case of $j=K$,
\begin{equation}
\label{eq:diff_iK(t)}
{i_K}'(t)= K\beta {p}i_{K-1}(t)s(t)-K\beta p i_{K}(t)s(t)-K(\gamma+\delta)i_{K}(t).
\end{equation}
And the corresponding initial conditions are 
\begin{equation}
\label{eq:initialcondition_s}
    s(0)=1-\varepsilon,
\end{equation}
\begin{equation}
\label{eq:initialcondition_i1}
     i(0)=i_1(0)=\varepsilon,
\end{equation}
and 
\begin{equation}
\label{eq:initialcondition_iK}
    i_2(0)=....=i_K(0)=0.
\end{equation}
Essentially, we check if we are allowed to use the Theorem 5.2 stated in \citep{andersson_stochastic_2000} to show the convergence of truncated density process $E^{(n)}_{K}(t).$
First of all, we notice that there are several jumps which can affect the process. In the case of a new non-to-be-reported infection, the process changes by 
$(-1,1,0,...,0)$
with the corresponding jump intensity function 
\begin{equation*}
    f_{(-1,1,0,...,0)}(s,i_1,i_2,...,i_K)=\beta(1-p)s\sum_{j=1}^K i_j. 
\end{equation*}
If there is a new to-be-reported infection comes to the component of size $1,$ then the process changes by 
$
    (-1,-1,2,0,...,0)
$
with the corresponding jump intensity function 
\begin{equation*}
    f_{(-1,-1,2,0,...,0)}(s,i_1,i_2,i_3...,i_K)=\beta p s i_1.
\end{equation*}
When there is a natural recovery comes to the component of size $1$ the process changes by 
$
    (0,-1,0,...,0)
$
with the corresponding jump intensity function 
\begin{equation*}
    f_{(0,-1,0,...,0)}(s,i_1,i_2,...,i_K)=\gamma i_1,
\end{equation*}
whereas if the whole component of size $1$ is diagnosed, the process would change by 
$
    (0,-1,0,...,0)
$
with the corresponding jump intensity function 
\begin{equation*}
    f_{(0,-1,0,...,0)}(s,i_1,i_2,...,i_K)=\delta i_1.
\end{equation*}
In the case when there is a new to-be-reported infection comes to the component of size $j=2,..,(K-1),$ the process changes by 
$
    (-1,0,...,0,-j,j+1,0,...,0)
$
with the corresponding jump intensity function 
\begin{equation*}
    f_{(-1,0,...,0,-j,j+1,0,...,0)}(s,i_1,...,i_{j-1},i_{j},i_{j+1},i_{j+2},...,i_{K})=\beta p s i_j.
\end{equation*}
Moreover, for the component of size $K,$ if there is a new to-be-reported infection occurs, then the process changes by 
$
    (-1,0,...,0,-K)
$
with the corresponding jump intensity function 
\begin{equation*}
    f_{(-1,0,...,0,-K)}(s,i_1,...,i_{K-1},i_{K})=\beta p s i_K.
\end{equation*}
For $j=2,..,K,$ if there is a natural recovery comes to the component of size size $j,$ then the process changes by 
$
    (0,0,...,0,j-1,-j,0,...,0)
$
with the corresponding jump intensity function 
\begin{equation*}
    f_{(0,0,...,0,j-1,-j,0,...,0)}(s,i_1,...,i_{j-2},i_{j-1},i_{j},i_{j+1},...,i_{K})=\gamma i_j,
\end{equation*}
Further, the process changes by 
$
    (0,0,...,0,0,-j,0,...,0)
$
if a component of size $j$ is diagnosed, the corresponding jump intensity function is given by
\begin{equation*}
    f_{(0,0,...,0,0,-j,0,...,0)}(s,i_1,...,i_{j-2},i_{j-1},i_{j},i_{j+1},...,i_{K})=\delta i_j.
\end{equation*}
Then we obtain the drift function $F$ defined in the Section 5.3 of \citep{andersson_stochastic_2000}
,which is here given by 
\begin{align*}
    F(s,i_1,i_2,...,i_{K}) 
    &= \begin{pmatrix}
          -\beta s i \\
           %\gamma i \\
           %\sum_{j=1}^{K}j\delta i_{j}\\
           \beta(1-p)si+\gamma i_{2}-\beta p s i_{1}-(\gamma +\delta)i_{1}\\
           2\beta {p}si_{1}+2\gamma i_{3}-2\beta p s i_{2}-2(\gamma+\delta)i_{2}\\
           \vdots \\ 
           %(K-1)\beta {p} s i_{K-2}+(K-1)\gamma i_{K}-(K-1)\beta p s i_{K-1}(t)-(K-1)(\gamma+\delta)i_{K-1}\\
           K\beta {p} s i_{K-1}-K\beta p s i_{K}-K(\gamma+\delta)i_{K}
         \end{pmatrix}.
  \end{align*}
It can be shown that for any $x=(s,i_1,i_2,...,i_{K})$ and $y=(s',i'_1,i'_2,...,i'_{K})$ in domain
$$C=\{x=(x_k)_k \in \mathbb{R}^{K+1}: 0\leq x_k\leq 1,k=1,..,K+1\}, $$ there exists a bound $M > 0$ such that
\begin{equation*}
    \vert F(x)-F(y) \vert \leq  M \vert x-y \vert,
\end{equation*}
with the absolute norm $\vert \cdot \vert$ in $\mathbb{R}^{K+1}.$ This bound $M$ can be roughly given by 
\begin{equation*}
    M=\max\{2\beta+2\beta p K^{2},2\beta+(2\beta p+2\gamma+\delta) K^{2}-\beta p\}.
\end{equation*}
Finally, we are now allowed to apply the Theorem 5.2 in \citep{andersson_stochastic_2000}, which showing that on any bounded intervals $[0,t_{end}],$ the truncated "density" process
$$E^{(n)}_{K}(t)=(S^{(n)}(t)/n,I^{(n)}_1(t)/n,I^{(n)}_2(t)/n,...,I^{(n)}_K(t)/n)$$ converges almost surely to the deterministic process 
$${E^{\infty}_K}(t)=(s(t),i_1(t),i_2(t),...,i_K(t)),$$ 
which is defined by Equation (\ref{eq:diff_s(t)})-(\ref{eq:initialcondition_iK}).

This convergence of the truncated processes combined with the earlier sketch of why the truncated processes approximate the infinite systems well by choosing $K$ large completes the proof of Theorem \ref{theorem:mainphase}.

\end{proof}
%%%%%%%%%%%%%%%%%%%%%%%%%%%%%%%%%%%%%%%%%%%%%%%%%%%%%%%
%%%%%%%%%%%%%%%  Numerical Illustrations %%%%%%%%%%%%%%%
%%%%%%%%%%%%%%%%%%%%%%%%%%%%%%%%%%%%%%%%%%%%%%%%%%%%%%%%%%%%
\section{Numerical Illustrations}
\label{sec:numerical}
\subsection{Original Model}
In this section we perform simulations supporting our large population results, and also investigate the effect of the TT-strategy. We do this mainly for the following parameter values (inspired from the Covid-19 pandemic). Before the TT-strategy is applied we have the Markovian SIR epidemic model with $\beta=0.75$ and $\gamma=0.25$, implying an average infectious period of $1/\gamma= 4$ days and a basic reproduction number $R_0=\beta/\gamma=3$. When the TT-strategy is considered fix, we assume that
 $\delta=0.125$ and $p=0.5$ implying that 1/3 of the infected individuals are tested and isolated while still infectious and that half of their contacts are reported for contact tracing (\citep{lucas_engagement_2020} believed that the fraction $p$ of contacts that were successfully traced varies between $40\%$ and $80\%$. 
Moreover in the following text, whenever computing the reproduction numbers $R^{(c)}_{*}$ and $R^{(ind)}_{*}$ numerically, we approximate the infinite sum in Equation (\ref{eq:expressionENc}) by a finite sum with truncation size of $100.$

First we performed 10~000 simulations of the epidemic and stored the final number infected in each simulation. We did this for three different population sizes, $n=$ 1000, 5000 and 10~000, each simulation starting with one initial infectious individual. We say (quite arbitrarily) that there is a minor outbreak when at most $10\%$ get infected during the outbreak, otherwise a major outbreak occurs. We summarize the fraction of minor outbreaks and the empirical mean fraction of infected individuals among the major outbreak cases in {Table \ref{tab:simulationresults}}. To these simulations we add a line for the limiting results (denoted by $n=\infty$). In this line we have derived the minor outbreak probability using Equation (\ref{eq:generatingfct}) with truncated sum up to $100$ and the mean fraction of the major outbreaks is computed numerically using Equations (\ref{eq:diff_s(t)})-(\ref{eq:initialcondition_iK}) with truncation size $K=100$ where $r_\infty$ is approximated by $r(t)$ for $t=100$~Days and $\varepsilon = 0.01$ which shows evidence that our limiting approximations work quite well already for these moderate population sizes. In particular, we observe from the second column of Table \ref{tab:simulationresults} that the mean fraction of the major outbreaks becomes closer to the deterministic limit $r_\infty$ for larger $n$ (see Conjecture \ref{conjecture}). As shown in Fig.~\ref{fig:h_1000_major},\ref{fig:h_5000_major} and \ref{fig:h_10000_major} the distribution for the major outbreaks is more peaked when the population is larger. We also note from those zoomed histograms for the major outbreaks (Fig. \ref{fig:h_1000_major},\ref{fig:h_5000_major} and \ref{fig:h_10000_major}), that they seem to follow a normal distribution with the deterministic limit as center, especially for larger $n$ (see Remark \ref{remark_CLT}). 
%%%%%%%%%%%%%%%%%%%%%%%%%%%%%%%%%%%%%%%%%%%%%%%%%%%%%%%
%%%%%%%%%%%%%%% epidemic in case of 1000,5000,10000 population %%%%%%%%%%%%%%%%%%%
% Histograms: (one figure with the three sizes below each other,)
%%%%%%%%%%%%%%%%%%%%%%%%%%%%%%%%%%%%%%%%%%%%%%%%%%%%%%%%%%%%%%%%%
\begin{figure}
\begin{subfigure}{0.5\textwidth}
\centering
\includegraphics[width=59mm]{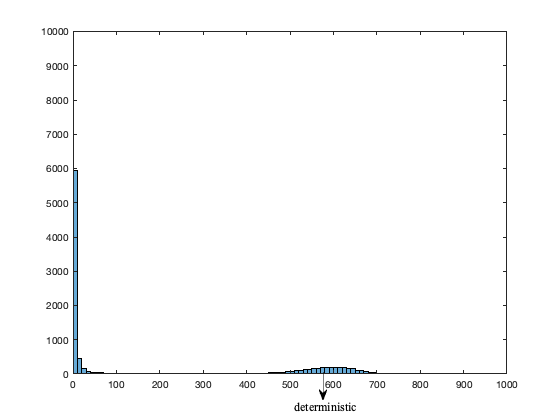}
\caption{}
\label{fig:h_1000_full}
\end{subfigure}
     \hfill
     \begin{subfigure}{0.5\textwidth}
     \centering
         \includegraphics[width=60mm]{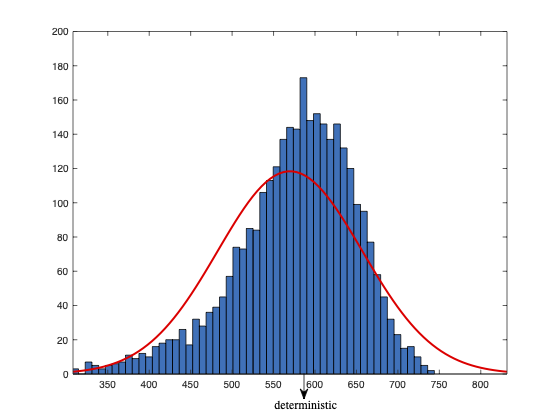}
         \caption{}
         \label{fig:h_1000_major}
     \end{subfigure}
     \hfill
     % Size: 5000
     \begin{subfigure}{0.5\textwidth}
\centering
\includegraphics[width=59mm]{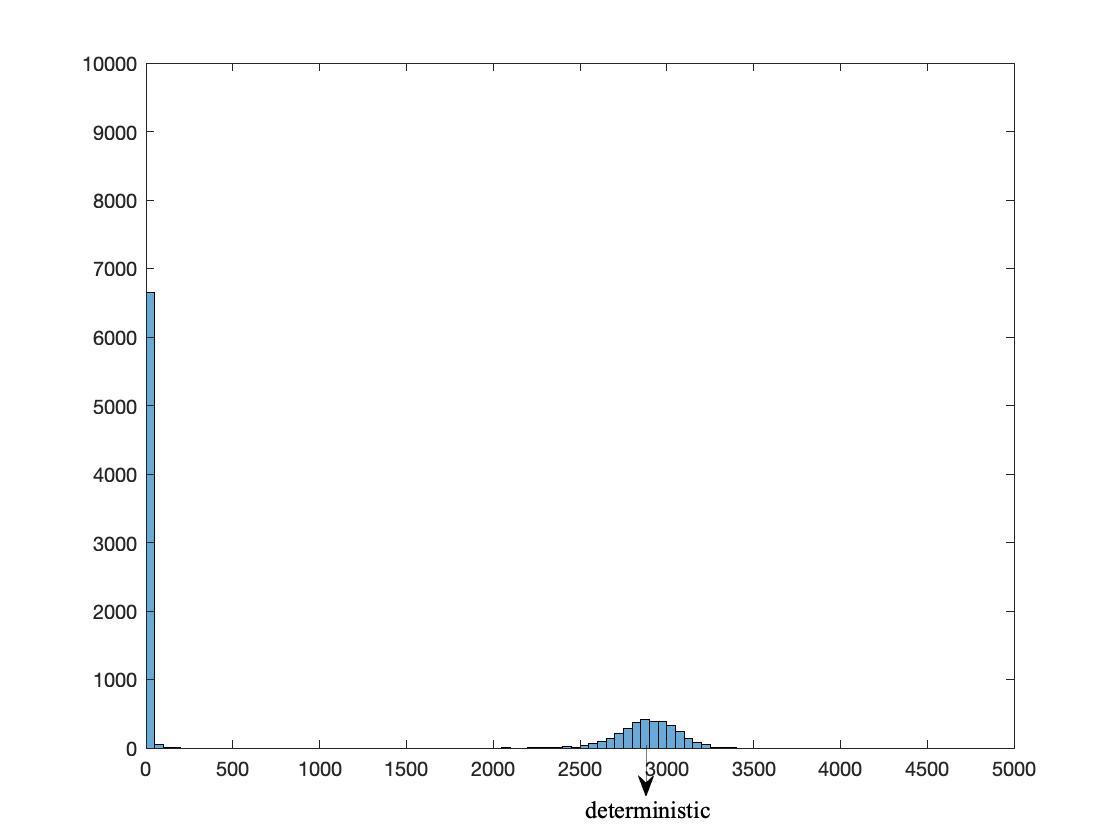}
\caption{}
\label{fig:h_5000_full}
\end{subfigure}
\hfill
     \begin{subfigure}{0.5\textwidth}
     \centering
         \includegraphics[width=60mm]{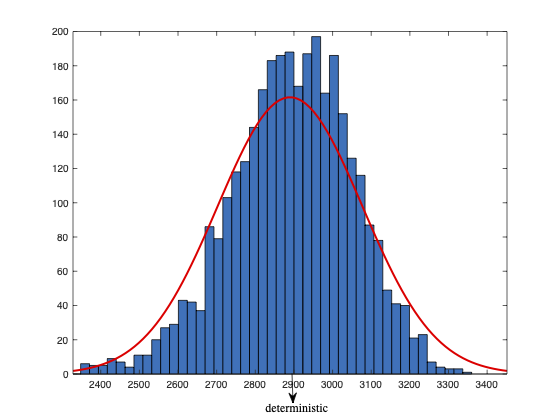}
         \caption{}
         \label{fig:h_5000_major}
     \end{subfigure}
     
    \begin{subfigure}{0.5\textwidth}
\centering
\includegraphics[width=59mm]{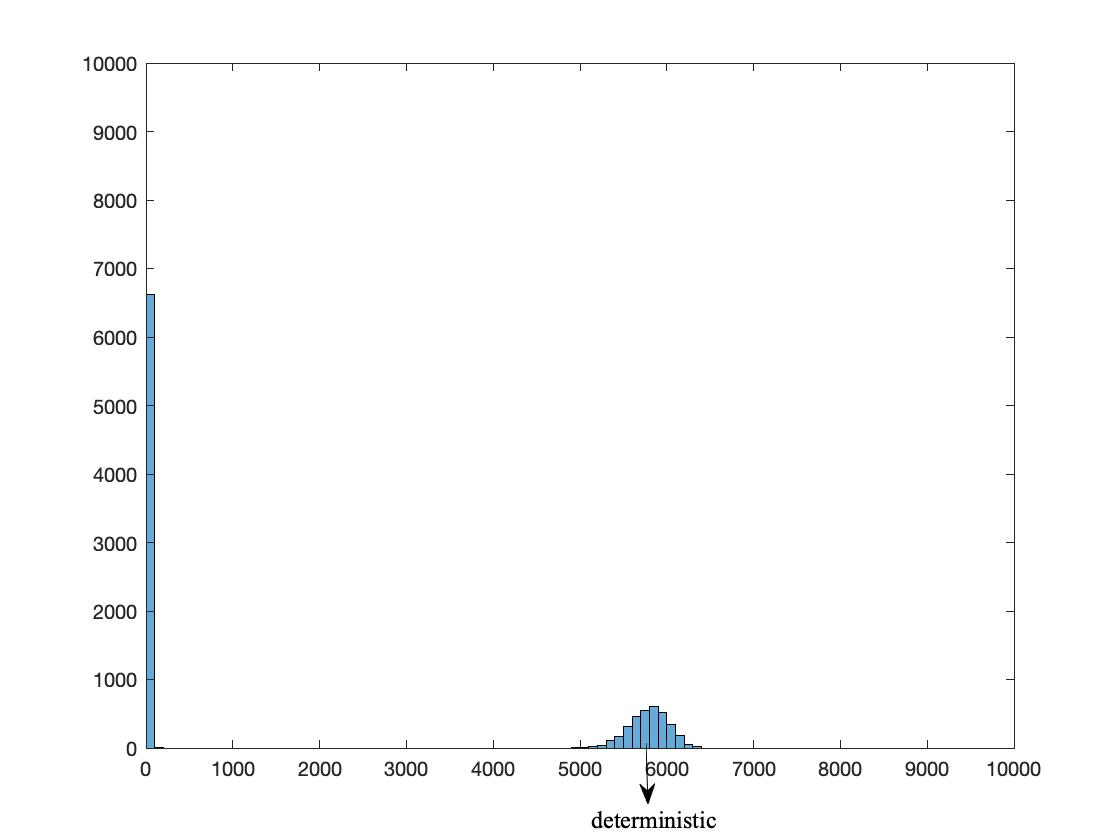}
\caption{}
\label{fig:h_10000_full}
\end{subfigure}
     \hfill
     \begin{subfigure}{0.5\textwidth}
     \centering
         \includegraphics[width=60mm]{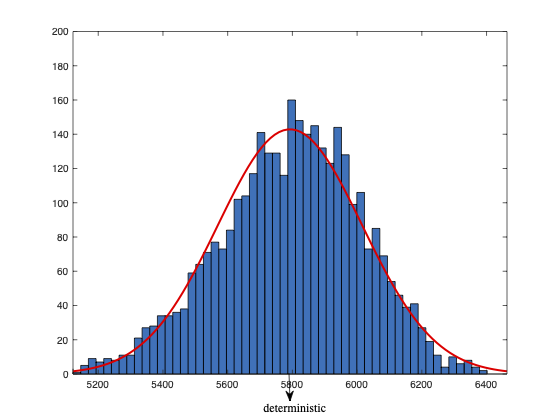}
         \caption{}
         \label{fig:h_10000_major}
     \end{subfigure}
     
     \caption{Histogram of the final size in 10 000 simulations of epidemic with population size in \ref{fig:h_1000_full}-\ref{fig:h_1000_major} $n=1000$, in \ref{fig:h_5000_full}-\ref{fig:h_5000_major} $n=5000$ and in \ref{fig:h_10000_full}-\ref{fig:h_10000_major} $n=10~000$, starting with one initial infective with full histogram to the left and zoomed in on the major outbreaks to the right with normally fitted curve in red}
	
        \label{fig:histogram}
\end{figure}
%%%%%%%%%%%%%%%%%%%%%%%%%%%%%%%%%%%%%%%%%%%%%%%%%%%%%%%
% Table: Simulation results about final fraction infected
%%%%%%%%%%%%%%%%%%%%%%%%%%%%%%%%%%%%%%%%%%%%%%%%%%%%%%%
\begin{table}[ht]
\caption{Simulation results about final fraction infected}
\label{tab:simulationresults}       % Give a unique label
\centering
\begin{tabular}{llll}
\hline\noalign{\smallskip}
size of & mean fraction of & standard & fraction of \\ 
  population & infected among major outbreaks & deviation & minor outbreaks \\ 
\noalign{\smallskip}\hline\noalign{\smallskip}
1000 & 0.5698 & 0.0873 & 0.6803 \\ 
5000 & 0.5786 & 0.0323 & 0.6707 \\ 
10000 & 0.5793 & 0.0224 & 0.6622 \\
$\infty$ & 0.5790 (=$r_\infty$)& 0 & 0.6667 \\
\end{tabular}
\end{table}
%%%%%%%%%%%%%%%%%%%%%%%%%%%%%%%%%%%%%%%%%%%%%%%%%%%%%%%
%%%%%%%%%%%%%%%threshold results R^(c) %%%%%%%%%%%%%%
%%%%%%%%%%%%%%%%%%%%%%%%%%%%%%%%%%%%%%%%%%%%%%%%%%%%%%%%

Next, we illustrate the threshold results saying that when $R^{(c)}_{*}\le 1$ we expect only minor outbreaks to take place whereas when $R^{(c)}_{*}> 1$ also large outbreaks may occur. We first fix the parameters $(\gamma,\delta,p)=(0.25,0.125,0.5),$ and choose the $\beta$ to be 0.40, 0.50, 0.59 and 0.67, so that the corresponding effective component reproduction number takes values of 0.75, 1.00, 1.25 and 1.50 using Equation (\ref{eq:expressionRc}). Then for each case of $R^{(c)}_{*}$, we did 10 000 simulations of the epidemic with fixed size of population $5000.$ In Table \ref{tab:simulationresults_R}, we show the fraction of minor outbreaks and the mean fraction of infected individuals among the major outbreaks. We see that in the case of \mbox{$R^{(c)}_{*} < 1$} there are nearly no major outbreaks, whereas more major outbreaks occur for \mbox{$R^{(c)}_{*} > 1$}.  As \mbox{$R^{(c)}_{*}$} grows bigger, there are larger outbreaks.

\begin{table}[ht]
\caption{Simulation results of epidemics with fixed $n=5000$}
\label{tab:simulationresults_R}       % Give a unique label
\centering
\begin{tabular}{llll}
\hline\noalign{\smallskip}
reproduction number & mean fraction of & standard & fraction of \\ 
$R^{(c)}_{*}$ & infected among major outbreaks & deviation & minor outbreaks \\ 
\noalign{\smallskip}\hline\noalign{\smallskip}
0.75 & 0.1315  & 0.0275 & 0.9945\\ 
1.00 & 0.1867 & 0.0666 & 0.9452\\ 
1.25 & 0.3037 & 0.1100 &  0.8442\\
1.50 & 0.4437 & 0.1150 &  0.7561  \\
\end{tabular}

\end{table}
%%%%%%%%%%%%%%%%%%%%%%%%%%%%%%%%%%%%%%%%%%%%%%%%%%%%%%%
%%%%%%%%%%%%%%% deterministic and stochastic %%%%%%%%%%%%%%
%%%%%%%%%%%%%%%%%%%%%%%%%%%%%%%%%%%%%%%%%%%%%%%%%%%%%%%%

We now study the time evolution of the epidemics showing that it becomes less random as population size $n$ increases, as stated in Theorem \ref{theorem:mainphase}. We do this by plotting random epidemic processes $\{I_n(t)/n\}$ and comparing it with the limiting deterministic process $\{i(t)\}$. 
As before, we use the parameter values $(\beta,\gamma,\delta,p)=(0.75,0.25,0.125,0.5)$. 
More specifically we plot the deterministic (in red) curves of the fraction of infectives when the population size is 1000, 5000 and 10 000 respectively. For each population size, we plot the fraction of infected for one simulation (in black), then we did 10 simulations given each size of population and plot the empirical mean of the fraction of infected (in blue). We can see from {Fig.~\ref{fig:fract_infect}} that the larger the population size, the better the truncated deterministic process approximates the epidemic process. All simulations were started with 1\% being infectious ($I_n(0)/n=0.01$) and the rest susceptible. The deterministic fraction of infectives are derived by solving Equations (\ref{eq:diff_s(t)})-(\ref{eq:initialcondition_iK}) with $\varepsilon = 0.01$ and truncation size $K=100$.
\begin{figure}
\begin{subfigure}{0.3\textwidth}
\centering
\includegraphics[width=39mm]{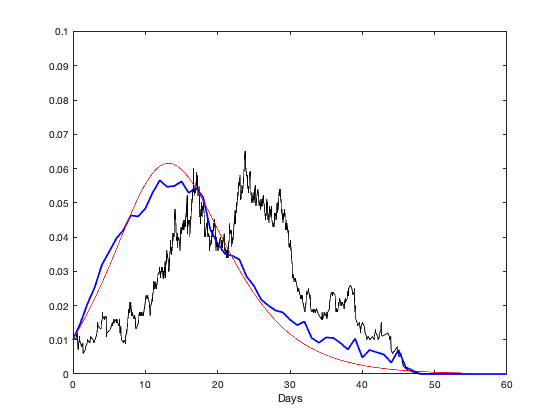}
\caption{}
\label{fig:I_1000}
\end{subfigure}
     \hfill
     \begin{subfigure}{0.3\textwidth}
     \centering
         \includegraphics[width=40mm]{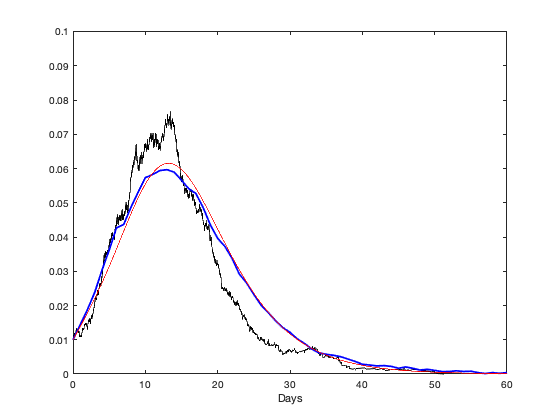}
         \caption{}
         \label{fig:I_5000}
     \end{subfigure}
     \hfill
     \begin{subfigure}{0.3\textwidth}
     \centering
         \includegraphics[width=40mm]{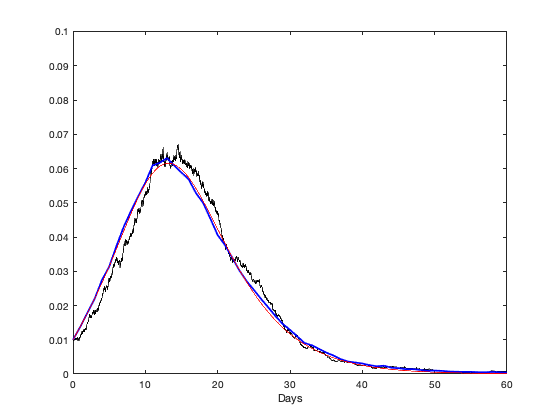}
         \caption{}
         \label{fig:I_10000}
     \end{subfigure}
        \caption[]{Fraction of infectives with population size of {\ref{fig:I_1000}} 1000, \ref{fig:I_5000} 5000 and {\ref{fig:I_10000}} 10 000  with $1\%$ initial infectives. 
        The fraction of infectious individuals for one stochastic simulation is in black, the one for deterministic is in red, whereas the empirical mean of ten simulations is in blue}
        \label{fig:fract_infect} % Give a unique label
\end{figure}
%%%%%%%%%%%%%%%%%%%%%%%%%%%%%%%%%%%%%%%%%%%%%%%%%%%%%%%
%%%%%%%%%%%%%%% effect of TT in original model %%%%%%%%%%%%%%
%%%%%%%%%%%%%%%%%%%%%%%%%%%%%%%%%%%%%%%%%%%%%%%%%%%%%%%%

Moreover, we investigate the effect of TT strategy. We recall that $\delta$ denotes the rate of testing (either broad screening or more targeted testing) and isolate those who test positive immediately, and $p$ denotes the fraction of all contacts of infectious individuals that are successfully contact traced. In {Fig.~\ref{fig:R_original}}, we plot the effective reproduction numbers $ R^{(c)}_{*}$ and $R^{(ind)}_{*}$ derived by Equations (\ref{eq:expressionRc}) and (\ref{eq:expressionRind}) respectively, as a function of the fraction of infectives being tested (before natural recovery) $\delta/(\delta+\gamma)$ in $[0, 0.5]$ and of $p$ in $[0, 1]$, keeping the other two parameters fixed at $\beta=0.75$ and $\gamma=0.25.$ {Fig.~\ref{fig:R_c_originalmodel}} shows that, surprisingly, $R^{(c)}_{*}$ is not monotone in $p$, whereas {Fig.~\ref{fig:R_ind_originalmodel}} shows that the individual reproduction number $R^{(ind)}_{*}$ seems to be, as expected. As seen from the contour lines where $R^{(ind)}_{*}=2.5, 2, 1.5, 1,$ the lines are steeper with lower $R^{(ind)}_{*}$. When it comes to comparing the effects of $p$ and $\delta/(\delta+\gamma)$ on $R^{(ind)}_{*},$ it seems as if $\delta/(\delta+\gamma)$ is more influential for high values (larger than 2.5 in this case) on $R^{(ind)}_{*}$, whereas for lower values (smaller than 2) on $R^{(ind)}_{*}$, tracing is more influential in preventing a major outbreak (i.e. reducing $R^{(ind)}_{*}$ below 1).  

\begin{figure}[h]
\begin{subfigure}[h]{0.5\textwidth}
\centering
\includegraphics[width=59mm]{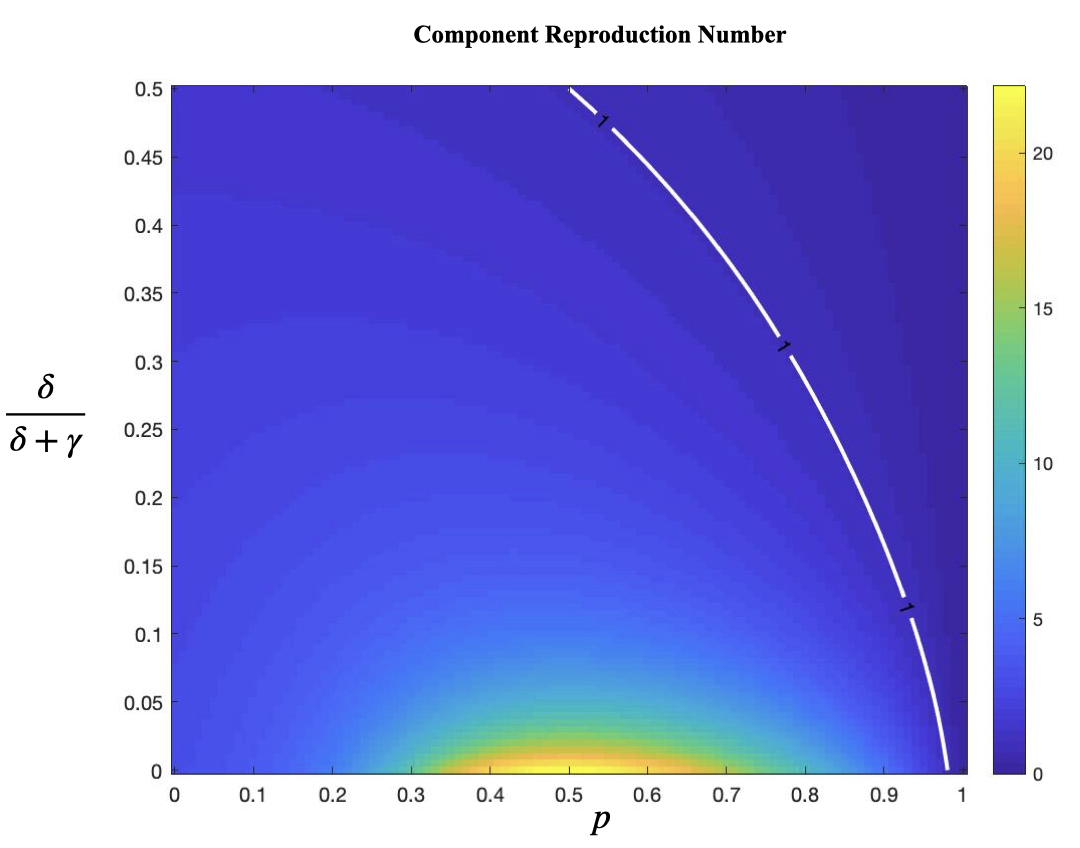}
\caption{}
\label{fig:R_c_originalmodel}
\end{subfigure}
     \hfill
\begin{subfigure}[h]{0.5\textwidth}
\centering
\includegraphics[width=59mm]{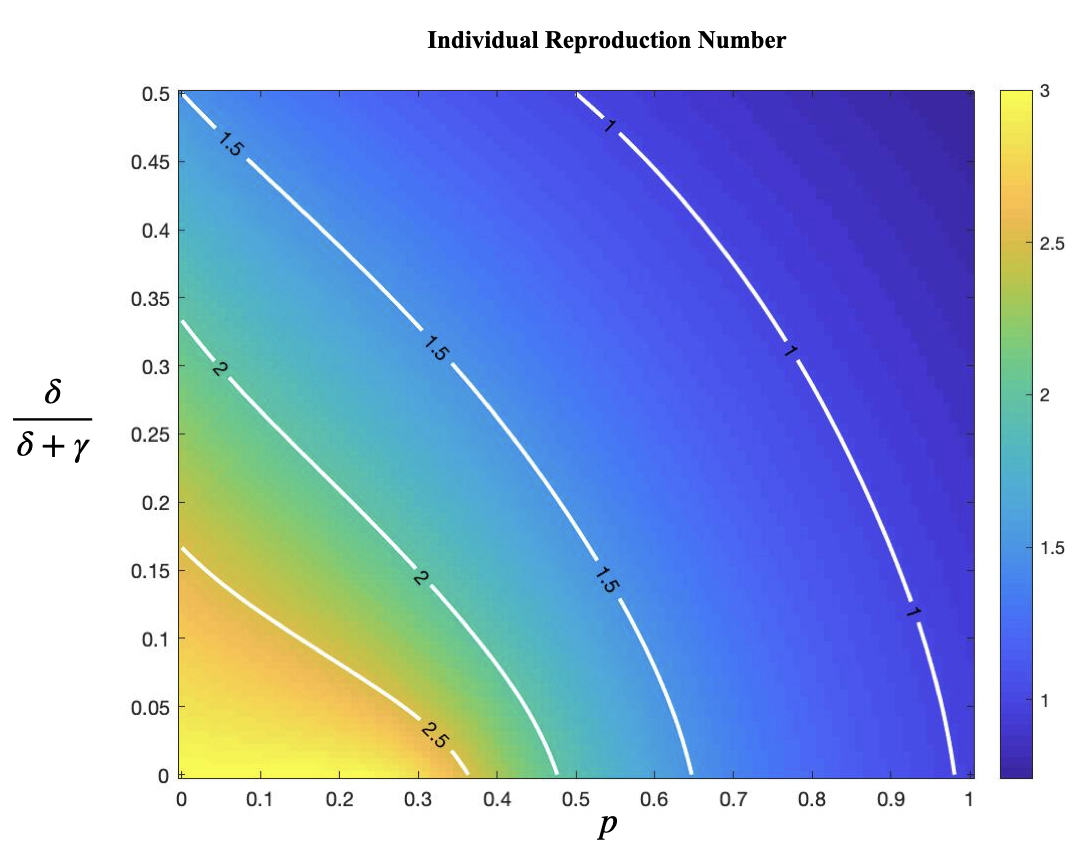}
\caption{}
\label{fig:R_ind_originalmodel}
\end{subfigure}
\caption{Heat map of the effective reproduction number, \ref{fig:R_c_originalmodel} for component $R^{(c)}_{*}$ and \ref{fig:R_ind_originalmodel} for individual $R^{(ind)}_{*}$, in a fine grid of $p=0, 0.05,..., 0.95, 1$ and $\delta/(\delta+\gamma)=0, 0.005,..., 0.495, 0.5.$ with $\beta=0.75, \gamma=0.25$ fixed. The white lines indicate where $R^{(c)}_{*}=1$ and $R^{(ind)}_{*}=2.5, 2, 1.5, 1,$ respectively
}
\label{fig:R_original}
\end{figure}
%%%%%%%%%%%%%%%%%%%%%%%%%%%%%%%%%%%%%%%%%%%%%%%%%%%%%%%
%%%%%%%%%%%%%%% Alternative Model %%%%%%%%%%%%%%%%%%%%%
%%%%%%%%%%%%%%%%%%%%%%%%%%%%%%%%%%%%%%%%%%%%%%%%%%%%%%%%

\subsection{Alternative model interpretation}
Finally, we turn to focus to the alternative model interpretation, where instead of $(\gamma,\delta)$, we have  $(\gamma,\nu+\delta)$ with $\gamma$ being the rate of natural recovery, $\nu$ the rate of self-reporting and $\delta$ the rate of screening. To start we fix $\beta=0.75$ and $\gamma=1/12, \nu=2\gamma=1/6$ implying that before screening ($\delta =0$), there would be $\nu/(\gamma+\nu) = 2/3$ of the infectious individuals get tested and self-isolated by own initiative. In Fig.~\ref{fig:R_newmodel}, we plot two reproduction numbers $R^{(c)}_{*}$ and $R^{(ind)}_{*}$ as a function of $p$ in $[0,1]$ and of the screening fraction $\delta/(\delta+\nu+\gamma)$ in $[0, 0.5]$. 

In the lower panels we show the corresponding heat maps, but now for the case $\gamma=0.2$ and $\nu=0.05$ implying that only 1/5 of infectives self-report, thus being closer to the original model where no individuals get tested prior to the introduction of screening. It is seen that whether the component reproduction number is monotone in $p$ or not depends on what fraction that self-report when having symptoms. Another difference as compared to the original model is that the tracing probability $p$ clearly has a bigger impact on reducing $R^{(ind)}_{*}$. An explanation to this would be that both self-tested individuals and those being screened will be contact traced.

Furthermore, if we assume that infectives who develop symptoms recover only due to diagnosis/self-reporting, then the fraction of asymptomatic individuals is exactly the fraction $\gamma/{(\gamma+\nu)}$ of infectives who do not self-report and are naturally recovered (without screening). By observing the steeper contour lines in Fig.~\ref{fig:R_ind_newmodel} with smaller fraction of asymptomatic infectives compared with that in Fig.~\ref{fig:R_ind_newmodel_gamma}, it implies that the tracing plays an even bigger role on reducing the individual reproduction number when there are larger fraction of individuals who are symptomatic.

\begin{figure}[h!]
   \begin{subfigure}{0.5\textwidth}
\centering
\includegraphics[width=59mm]{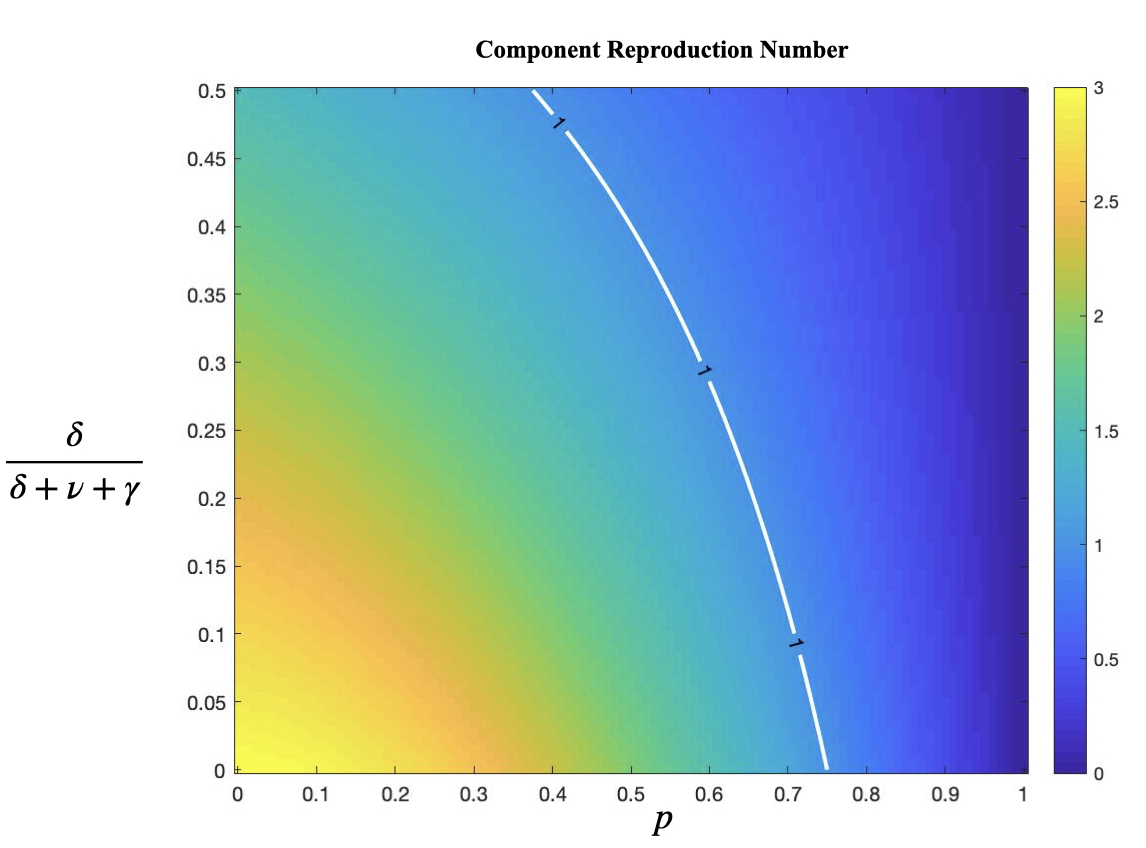}
\caption{}
\label{fig:R_c_newmodel}
\end{subfigure}
     \hfill
     \begin{subfigure}{0.5\textwidth}
\centering
\includegraphics[width=59mm]{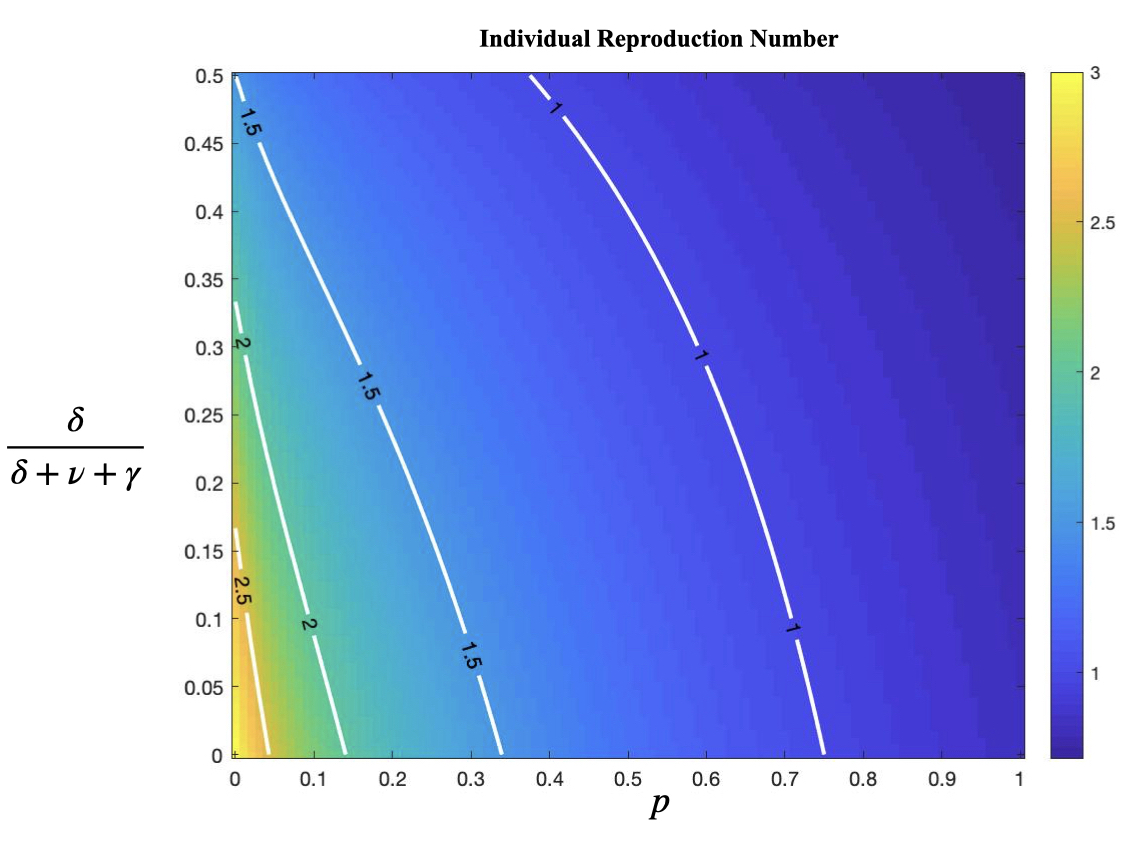}
\caption{}
\label{fig:R_ind_newmodel}
\end{subfigure}
\hfill
\begin{subfigure}{0.5\textwidth}
\centering
\includegraphics[width=59mm]{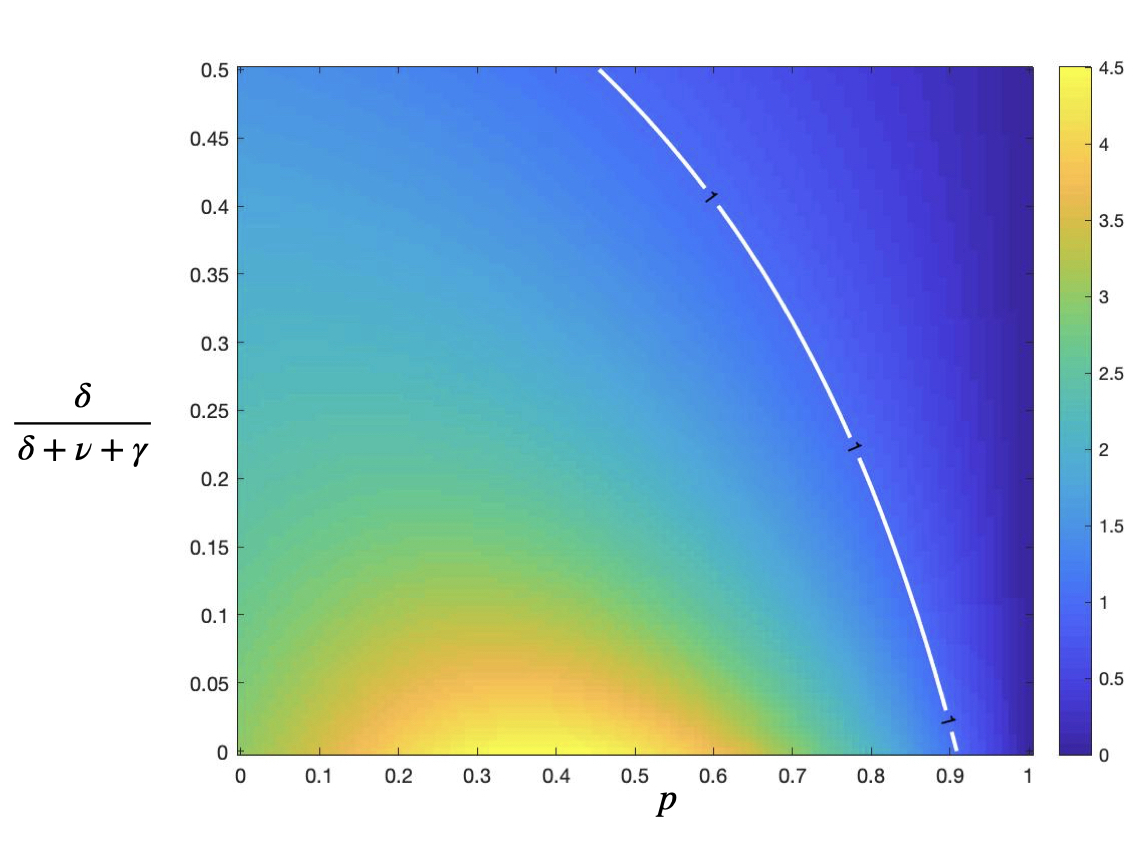}
\caption{}
\label{fig:R_c_newmodel_gamma}
\end{subfigure}
\begin{subfigure}{0.5\textwidth}
\centering
\includegraphics[width=59mm]{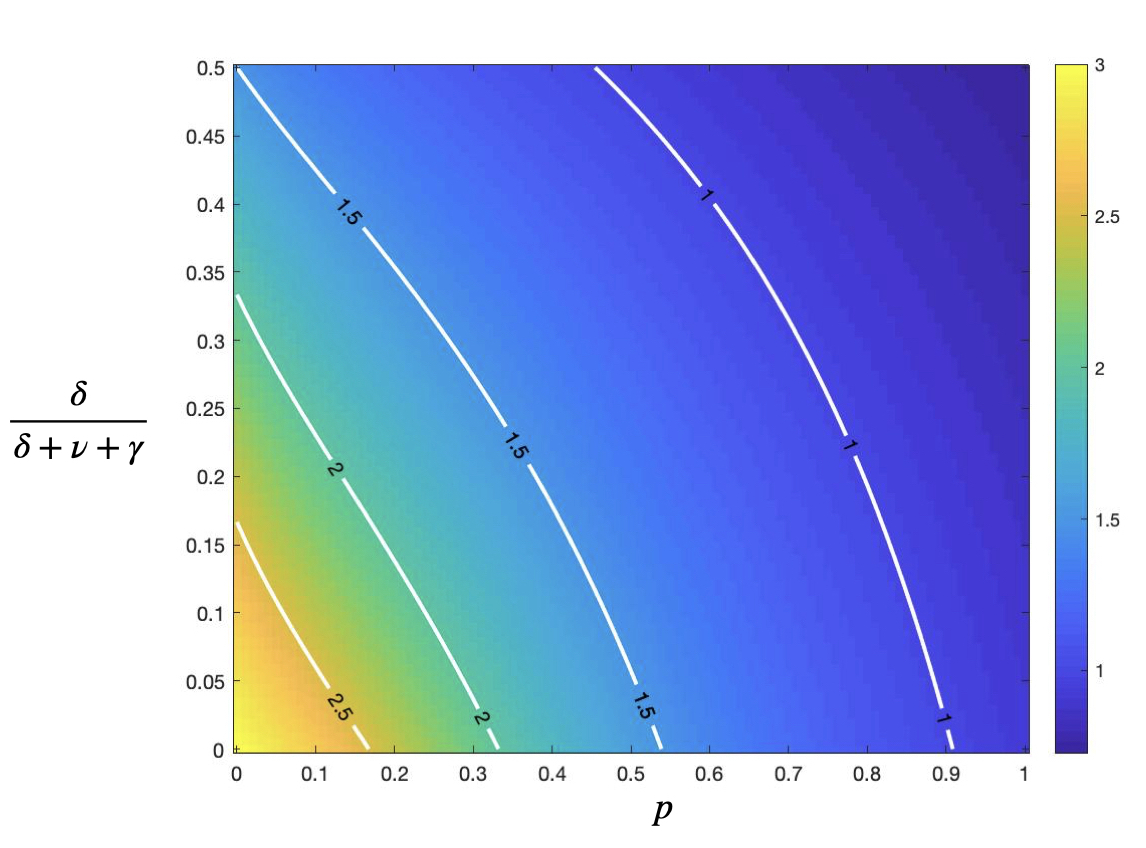}
\caption{}
\label{fig:R_ind_newmodel_gamma}
\end{subfigure}

\caption{Heat map of the effective reproduction numbers in the alternative model interpretation. Left panels are for $R^{(c)}_{*}$ and right panels for $R^{(ind)}_{*}$, varying  $ p$ in $[0,1] $ and $\delta/(\delta+\nu+\gamma)$ in $[0,0.5]$ with $\beta=0.75$ being fixed. In \ref{fig:R_c_newmodel} and \ref{fig:R_ind_newmodel} assuming that $\gamma=1/12$ and $\nu=1/6$ whereas in \ref{fig:R_c_newmodel_gamma} and \ref{fig:R_ind_newmodel_gamma} $\gamma=0.2$ and $\nu=0.05$.
The white lines show where $R^{(c)}_{*}=1$ in \ref{fig:R_c_newmodel}, \ref{fig:R_c_newmodel_gamma} and $R^{(ind)}_{*}=2.5, 2, 1.5, 1$ in \ref{fig:R_ind_newmodel}, \ref{fig:R_ind_newmodel_gamma} respectively
} 
\label{fig:R_newmodel}
\end{figure}

%%%%%%%%%%%%%%%%%%%%%%%%%%%%%%%%%%%%%%%%%%%%%%%%%%%%%%%%
%%%%%%%%%%%%%%% Conclusions and Discussion %%%%%%%%%%%%%
%%%%%%%%%%%%%%%%%%%%%%%%%%%%%%%%%%%%%%%%%%%%%%%%%%%%%%%%
\section{Conclusions and Discussion}
\label{sec:discussion}

In the paper we have analysed a Markovian epidemic model also incorporating the effect of testing and contact tracing (the Markovian SIR-TT-model). By analysing the process of to-be-reported components, rather than individuals themselves, it was shown that the early stage of the epidemic could be approximated by a suitable branching process, and that if an epidemic takes off, its behaviour becomes less random as the population size $n$ increases. The reproduction numbers, both for the components as well as for the individuals, were derived. Their dependence on the amount of testing and effectiveness of contact tracing were evaluated analytically as well as numerically. It was observed that the tracing probability $p$ had a bigger impact on reducing the individual reproduction number as compared to the fraction being tested through screening, and this difference was even more pronounced in the situation when some infectives self-test also without being screened (the alternative model interpretation). Surprisingly, the reproduction number for the components was not monotonically decreasing in $p$, but the individual reproduction numbers seem to be (as expected).

There are several possible extensions to the model making it more realistic. For instance, the model assumes that there are no delays in either contact tracing or testing. The results in the present paper can hence be seen as a best possible scenario, but allowing for a delay would of course give information on how important such delays are and how much would be gained if contact tracing would be quicker. Further, we make the simplifying 
assumption that traced individuals who have by then recovered are also
contact traced (cf. \citep{muller_contact_2000} does not make 
this assumption). Further, the model assumes no latent period and that the infectious period follows an exponential distribution. Introducing a latent period most likely makes testing and contact tracing more effective in that individuals may get screened as well traced before even becoming infectious, but how to quantify this effect remains to be analysed. A different step towards realism would be to consider a structured community as opposed to the current assumption of a uniformly mixing community. Such structure could for example include households, spatial aspects, or some other network structures.

One the other hand, as the model was defined, only contacts that resulted in infection are considered for contact tracing. In reality it may of course also happen that contacts that did not result in infection are reported and traced. During the early phase of an epidemic such tracing events will rarely find new infected cases, but later in the epidemic when transmission is extensive it could (the individuals may have been infected by other individuals). Similarly, we do not consider contact tracing if an infectious 
individual has close contact with an individual who has already been 
infected, since such contact does not result in infection. To allow also for these type of contact tracing events is much harder to analyse and remains an open problem. 

On the mathematical side two conjectures deserve to be proven (or disproved). The first is the statement for the final size of the epidemic in case of a major outbreak starting with one infective (see Conjecture~\ref{conjecture}). As in many similar epidemic models it seems highly plausible that this limiting final size agrees with that of the deterministic process taking $t$ to infinity and looking at ever smaller initial starting fractions $\varepsilon$, but a proof of this is missing. In addition, Fig. \ref{fig:histogram} shows that the final size seems to follow a normal distribution around the deterministic limit for the case where there is a major outbreak. Then we suggest that a related central limit theorem could be an open problem to be shown. Further and perhaps a lower hanging fruit, is to compute the proper effective individual reproduction number (see Remark \ref{remark_ind}) or prove that the individual reproduction number $R^{(ind)}_{*}$ in this paper is the correct one and it is monotonically decreasing both in $p$ and testing fraction $\delta/(\delta+\gamma)$. 

From an applied point of view it is of course important to have parameter estimates in order to say something quantitatively useful. The model has four parameters: $(\beta, \gamma, \delta, p)$. The average infectious period $1/\gamma$ is quite often known from earlier studies, and when the basic reproduction number $R_0=\beta /\gamma$ is known, estimates of $\beta$ would also be available. Nevertheless, the test-and-trace parameters $\delta$ and $p$ may be harder to estimate. In the case that testing comes from general broad screening it could be very well available: if for instance $1\%$ of the community is tested each day would lead to $\delta =0.01$ with day as time unit. If testing is targeted towards suspected cases it might be harder to know the rate $\delta$ at which infectious people are tested. Finally, estimates of the fraction $p$ of all infectious contacts that were detected by contact tracing is often hard to obtain. Perhaps a rough estimate could be obtained from studies investigating different type of contacts and how many infections they are responsible for. There are some statistical methods developed to estimate the tracing probability, e.g. a maximum-likelihood estimator in \citep{muller_estimating_2007} and an approximate Bayesian computation in \citep{blum_hiv_2010}. When it comes to digital contact tracing (by means of mobile tracing apps), the tracing probability $p$ would approximately correspond to the square of the app-using fraction. With higher app adoption, app contact tracing is expected to be more effective as compared with the traditional contact tracing, potentially due to the quicker identification and notification of infectious contacts (see e.g. \citep{Jenniskense050519,ferretti_quantifying_2020}). 

Analyses of epidemic models incorporating various preventive measures, and statistical studies relating to them, remains a research area deserving more attention in the future.

\bmhead{Acknowledgments}
T.B. is grateful 
to the Swedish Research Council (grant 2020-04744) for financial support.

\section*{Declarations}
\begin{itemize}
\item Funding:\newline
See acknowledgements.
\item Conflict of interest: \newline
We declare that we have no conflict of interest.
\item Ethics approval: \newline
Not applicable.
\item Consent to participate: \newline
Not applicable.
\item Consent for publication: \newline
Not applicable.
\item Availability of data and materials: \newline 
The paper uses no data or material.
\item Code availability: \newline
D.Z. used Matlab for all simulations and numerical illustrations. Codes are available upon request. 
\item Authors' contributions: \newline
T.B. initiated the project, the theory was developed jointly, D.Z. performed simulations and numerical computations. D.Z. wrote most of the first draft and both revised the manuscript.
\end{itemize}

\noindent

\bibliography{main}% common bib file

\end{document}